\documentclass[runningheads]{llncs}
\usepackage[T1]{fontenc}
\usepackage{graphicx}
\usepackage{breakcites}
\usepackage[colorlinks=true,breaklinks=true]{hyperref}
\usepackage{url}
\usepackage{color}

\usepackage{booktabs}   %
\usepackage{subcaption} %
\usepackage{listings}
\usepackage{relsize}
\usepackage{enumerate}
\usepackage{float}
\usepackage{nicefrac}
\usepackage{multirow}
\usepackage{microtype}
\usepackage{wrapfig}
\usepackage{amsmath}
\usepackage{amssymb}
\usepackage{mathtools}
\usepackage{xcolor}

\usepackage[algosection,ruled,vlined,linesnumbered,algo2e]{algorithm2e}
\usepackage{wrapfig}
\usepackage{algorithmic}
\usepackage[framemethod=tikz]{mdframed}
\usepackage{gensymb}
\usepackage[tableposition=below]{caption}
\usepackage{bussproofs}
\newcommand{\bracket}[1]{\langle #1 \rangle}  

\usepackage[subtle]{savetrees}

\definecolor{blue2}{rgb}{0.0, 0.5, 1.0}

\newif\ifdraft
\drafttrue %

\ifdraft
    \newcommand{\ravi}[1]{\textit{\color{blue2}[Ravi: #1]}}
    \newcommand{\corina}[1]{\textit{\color{blue2}[Corina: #1]}}
    \newcommand{\todo}[1]{\textit{\color{red}[#1]}}
    
\else
    \newcommand{\ravi}[1]{}
    \newcommand{\corina}[1]{}
    \newcommand{\todo}[1]{}
\fi

\newcommand{\paths}{\texttt{Paths}}
\newcommand{\unsafe}{\texttt{Unsafe}}
\newcommand{\safe}{\texttt{Safe}}

\newcommand{\ppsysmdp}{\Gamma_{Perf}}
\newcommand{\shieldppsysmdp}{\Gamma_{Perf}^{\Delta}}
\newcommand{\shieldcpsysmdp}{\Gamma_{Conf}^{\Delta}}
\newcommand{\shieldcpprsysmdp}{\overline{\Gamma}_{Conf}^{\Delta}}
\newcommand{\sysmdp}{\Gamma}

\DeclareMathOperator{\argmin}{argmin}

\definecolor{prismgreen}{rgb}{0, 0.6, 0}

\lstdefinelanguage{Prism}{ %
basicstyle=\color{red}\scriptsize\ttfamily, %
keywords={bool,C,ceil,const,ctmc,double,dtmc,endinit,endmodule,endrewards,endsystem,F,false,floor,formula,G,global,I,init,int,label,max,mdp,min,module,nondeterministic,P,Pmin,Pmax,prob,probabilistic,R,rate,rewards,Rmin,Rmax,S,stochastic,system,true,U,X},
keywordstyle={\bfseries\color{black}},
numberstyle=\tiny\color{black},
comment=[l] {//}, morecomment=[s]{/*}{*/}, %
commentstyle= \color{prismgreen}, %
tabsize=4, %
captionpos=b, %
escapechar=@, %
literate={->}{$\rightarrow{}$}{1}
}

\makeatletter
\newcommand{\printfnsymbol}[1]{%
  \textsuperscript{\@fnsymbol{#1}}%
}
\makeatother

\begin{document}

\title{Conformal Safety Shielding\\ for Imperfect-Perception Agents}

\titlerunning{Conformal Safety Shielding for Imperfect-Perception Agents}
\author{William Scarbro\thanks{Equal contribution}\inst{1} \and
Calum Imrie\printfnsymbol{1}\inst{2} \and
Sinem Getir Yaman\inst{2} \and
Kavan Fatehi\inst{2} \and \\
Corina P\u{a}s\u{a}reanu\inst{3} \and
Radu Calinescu\inst{2} \and
Ravi Mangal\thanks{Corresponding author (ravi.mangal@colostate.edu)}\inst{1}
}
\authorrunning{W. Scarbro, C. Imrie et al.}
\institute{Colorado State University, USA \and
University of York, UK \and
Carnegie Mellon University, USA}
\maketitle              %
\begin{abstract}

We consider the problem of safe control in discrete autonomous agents that use learned components for \emph{imperfect} perception (or more generally, state estimation) from high-dimensional observations. 
We propose a \emph{shield} construction that provides run-time safety guarantees  under perception errors by restricting the actions available to an agent, modeled as a Markov decision process, as a function of the state estimates. Our construction uses \emph{conformal prediction} for the perception component, which guarantees that for each observation, the predicted set of estimates includes the actual state with a user-specified probability. The shield allows an action only if
it is allowed for {\em all} the estimates in the predicted set, %
resulting in \emph{local safety}.
We also articulate and prove a \emph{global safety} property of existing shield constructions for perfect-perception agents bounding the probability of reaching unsafe states if the agent always chooses actions prescribed by the shield. 
We illustrate our approach with a case-study of an experimental autonomous system that guides airplanes on taxiways using high-dimensional perception DNNs.

\end{abstract}
\setcounter{footnote}{0} 

\section{Introduction}
\label{sec:intro}

The computational capabilities unlocked by neural networks have made it feasible to build autonomous agents that use learned neural components to interact with their environments for achieving complex goals.
As an example, consider an autonomous airplane taxiing system that senses the environment through a camera and is tasked with following the center-line on taxiways~\cite{KadronGPY21,fremont2020formal}.
The agent \emph{observes} the environment through sensors, \emph{perceives} its underlying state with respect to its environment (position on the taxiway) from the sensor readings via neural components, \emph{chooses} actions (go left, right, or straight) based on the perceived state using either symbolic or neural components, and \emph{executes} the actions to update the agent state.
The use of neural components for perception can cause the agents to make mistakes when perceiving the underlying state from the sensor readings. We refer to this phenomenon as \emph{imperfect perception}. 

Autonomous agents with imperfect perception are often intended to be deployed in safety-critical settings. Accordingly, one would like to ensure that all actions taken by the agent come with safety guarantees. 
Shielding~\cite{jansen2019safe,Konighofer2020LJB} is a promising technique that restricts the set of actions available to the agent at run-time in a manner that guarantees safety. Existing methods, however, either assume that perception is \emph{perfect}~\cite{jansen2019safe} or they require mathematically modeling the complex processes of generating observations and perceiving them~\cite{carrNJT2023safe,sheng2024safe}.

\paragraph{\textbf{MDP Formulation.}} We model agents and their environments together as Markov Decision Processes (MDPs) with safety expressed as a specification in a probabilistic temporal logic. Our MDP formulation is generic and can model any discrete agent that follows the observe-perceive-choose-execute loop while exhibiting stochasticity and non-determinism. Our formulation accounts for the fact that the agent can only access the underlying state through stochastic observations by defining a state of the MDP as a triple of the actual agent state, the observation corresponding to the state, and the estimate of the actual state.

\paragraph{\textbf{Safety through Conformal Perception and Shielding.}} In this work, we enforce safety for imperfect-perception agents by constructing, either offline~\cite{jansen2019safe} or online~\cite{konighofer2023online}, a \emph{conformal shield} that restricts 
the set of available actions in each state of the MDP as a function of the perceived or estimated states. The shield is independent to the actual state and observations resulting from it.
Importantly, constructing the shield does not require us to model either the complex generative process mapping actual states to observations or the perception process that maps observations to estimates. Our key insight is to leverage existing methods~\cite{pranger2021tempest,jansen2019safe} to construct shields for perfect-perception MDPs, i.e., MDPs modeling agents that are able to perceive the underlying states from observations without any errors (thereby making it unnecessary to model the observation and perception processes). Given such a shield, our next idea is to \emph{conformalize}~\cite{angelopoulos2021gentle,vovk2005algorithmic} the neural perception component so that, given an observation (for instance, an image), it predicts a set of state estimates instead of a single estimate of the underlying actual state. This set is guaranteed to contain the actual state with a user-specified probability, assuming that training and future data are \emph{exchangeable}---a weaker variant of the i.i.d. assumption. 
Our final step is to use the perfect-perception shield to compute the set of actions that it deems safe with respect to all the states in a predicted set. Since the actual state is guaranteed to be in the predicted set with a user-specified probability (due to conformalization),
the resulting set of actions comes with
a notion of safety, even under imperfect perception.

\begin{figure}[t]
\centering
\includegraphics[width=\columnwidth]{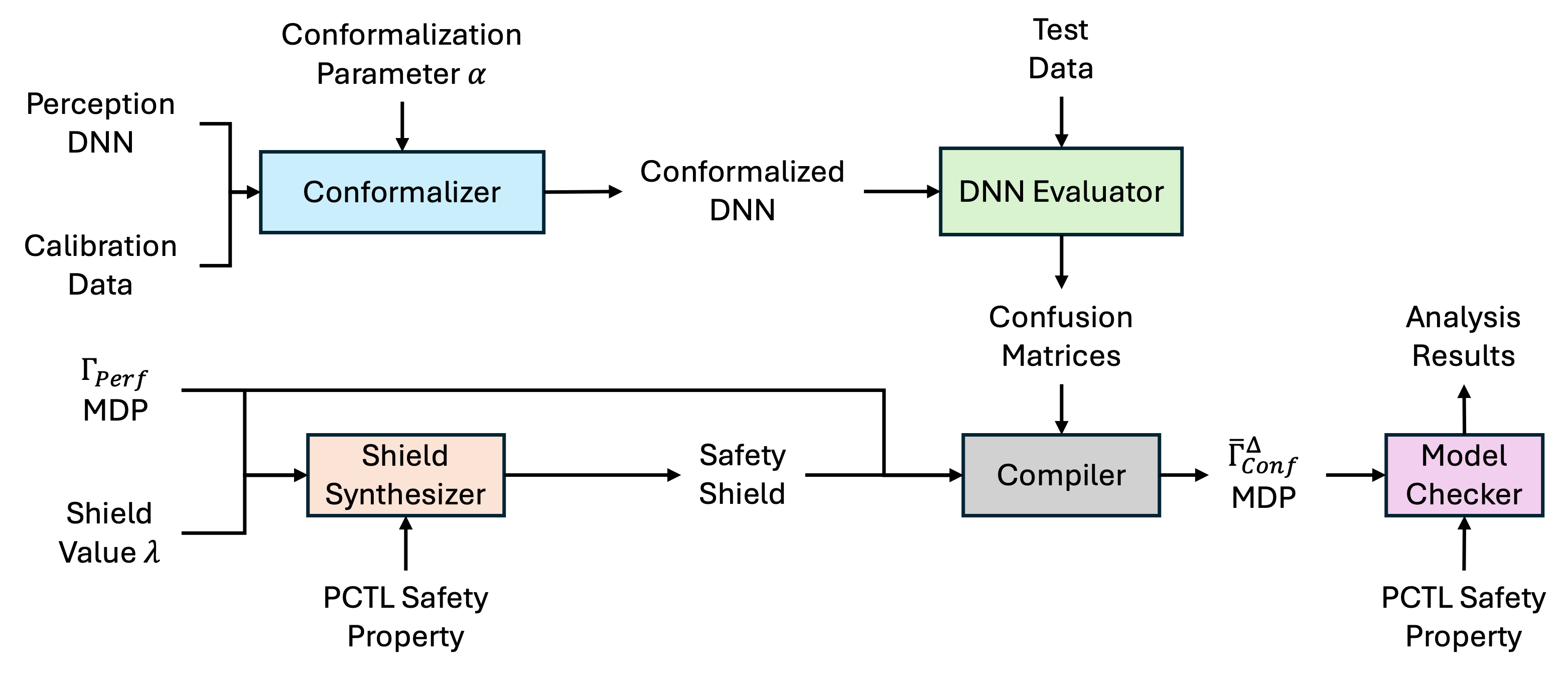} 
\caption{Overall pipeline for our approach. $\ppsysmdp$ refers to the MDP formulated for the perfect-perception agent (Section~\ref{sec:PP_MDP_Shielded}) while $\shieldcpprsysmdp$ refers to the MDP for the imperfect-perception agent with the conformalized shield and a probabilistic abstraction of the observation and perception processes (Section~\ref{sec:Abs_MDP_Shielded}).  }
\label{fig:pipeline}
\end{figure}

Figure~\ref{fig:pipeline} gives an overview of our approach. A \emph{conformalizer} constructs the conformalized perception deep neural network (DNN) from the original perception DNN using calibration data in the form of (actual state, observations) pairs. A \emph{shield synthesizer} is used to construct a shield for the perfect-perception agent modeled as the MDP $\ppsysmdp$. The synthesized shield and the conformalized perception can then be combined to construct our conformal shield suitable for imperfect-perception agents. Although our shield comes with a local safety guarantee, in our empirical evaluation, we also measure the overall safety of the shielded imperfect-perception agent via model checking. To do so, we formulate an MDP ($\shieldcpprsysmdp$) representing the shielded, imperfect-perception agent. Due to the complexity of the observation and perception processes, they are replaced in $\shieldcpprsysmdp$ with a probabilistic abstraction that maps actual states to a distribution over sets of estimated states. This abstraction is constructed using existing methods~\cite{calinescu2023IMRPSV,PasareanuMGYICY23} that consume a confusion matrix reporting the performance of the perception component. The \emph{compiler} puts all the pieces together to generate the $\shieldcpprsysmdp$ MDP from the $\ppsysmdp$ MDP, the perfect-perception shield,  and the confusion matrices representing the performance of the conformalized perception.

\paragraph{\textbf{Safety Guarantees.}} Our shield construction 
provides finite-horizon, probabilistic safety properties
where safety is characterized by a set of states to be avoided. 
In particular, our construction aims for a notion of \emph{local safety}---from any state of the imperfect-perception MDP, if one takes actions as prescribed by the shield, then from the resulting successor states, there very likely exist policies\footnote{A map from states to actions.} (or strategies or controllers) such that finite-horizon probabilistic safety is guaranteed, even though the shield does not have access to the actual state.
We also precisely express and prove the \emph{global safety} guarantees granted by existing shielding methods for perfect-perception MDPs when considering finite-horizon, probabilistic safety properties. A global guarantee ensures that, starting from the initial state, if the agent always chooses actions prescribed by the shield, then the probability of reaching unsafe states in a finite number of steps can be bounded. Our proof also highlights that, in general, shields can cause agents---even under perfect perception---to reach \emph{stuck} states, i.e., states where the shield is empty. 

\paragraph{\textbf{Case Study.}} We implement our approach and evaluate it using the case study of autonomous airplane taxiing system. 
We also evaluate the effect of different hyperparameters associated with conformalization and shielding on the safety outcomes. %
We show that our shield can reduce the probability of imperfect-perception agents reaching unsafe states compared to shields constructed assuming perfect perception. 

\paragraph{\textbf{Contributions.}} In summary, we make the following contributions: (i) a formal model of imperfect-perception agents as MDPs; (ii) a shield construction for imperfect-perception MDPs that does not require a precise model of the observation and perception processes
(iii) statement and proof of global safety guarantees granted by existing shield constructions for perfect-perception agents; (iv) an implementation and evaluation of our ideas via a case study.

\section{Theoretical Foundations}
\label{sec:theory}
In this section, we give MDP formulations of discrete autonomous agents with perfect as well as imperfect perception, with and without shields. We also state and prove local and global safety guarantees for the shielded agents.

\subsection{Imperfect-Perception MDP}

In this work, an autonomous system (or agent) is a tuple $(S,O,A,\iota,\rho,\eta,\omega)$ where $S$ is the finite system state space, $O$ is the finite observation space (for instance, the space of images of a fixed size), $A$ is the finite action space, $\iota \in S$ is the initial system state,  $\omega:S \times O \rightarrow [0,1]$ is the observation function that maps actual states to observations, $\rho:O\rightarrow S$ is the perception function mapping observations to state estimates, and $\eta:S\times A\times S \rightarrow [0,1]$ describes the system dynamics. For technical convenience, we assume $O$ has a \emph{null} element (denoted as $\bot_O$).

The system can be formulated as a Markov Decision Process (MDP), $\sysmdp := (\hat{S},A,\hat{\iota}, P)$, where $\hat{S} := S \times O \times S$ is the set of MDP states, $\hat{\iota} := \iota \times \bot_O \times \iota$ is the initial MDP state, and $P:\hat{S}\times A \times \hat{S}\rightarrow [0,1]$ is the (partial) probabilistic transition function defined as:
\begin{equation}
\begin{array}{ll}
P := & \forall s,s',s'' \in S, \forall o,o' \in O. \\
& (s,o,s') \xrightarrow[]{a' \in A,~\eta(s,a',s'')\cdot\omega(s'',m')} (s'',m',\rho(m'))
\end{array}
\end{equation}
where $\hat{s} \xrightarrow[]{a,~p}\hat{s}'$ denotes $P(\hat{s},a,\hat{s}'):=p$.
A controller (or policy or strategy) $\pi:\hat{S} \rightarrow A$ fixes the action to take in each state of the MDP. Given a policy $\pi$ for an MDP $\sysmdp$, the set of all possible system \emph{paths}, starting from a state $\hat{s}$ is given by
$\paths_{\sysmdp,\pi}(\hat{s}):=\{(\hat{s}_t)_{t=0}^{\infty} \in \hat{S}^*~|~\forall t\in\mathbb{N}.~P(\hat{s}_t,\pi(\hat{s}_t),\hat{s}_{t+1})>0~\wedge~\hat{s}_0=\hat{s}\}$. $\sysmdp$ with a fixed $\pi$ gives rise to a probability space over the set $\paths_{\sysmdp,\pi}(\hat{s})$ of system paths from the state $\hat{s}$. We use $\mathbb{P}_{\sysmdp,\pi}^{\hat{s}}$ to denote the probability measure associated with this space.
The probability of a finite path $p \in \paths_{\sysmdp,\pi}(\hat{s})$ consisting of states $(\hat{s}_0,...,\hat{s}_{n-1})$ may be calculated from the transition function $P$ and the policy $\pi$ as, 
\begin{equation}
\label{eq:path_prob}
    \mathbb{P}_{\sysmdp,\pi}^{\hat{s}}(p) = \prod_{t=0}^{n-2} P(\hat{s}_t,\pi(\hat{s}_t),\hat{s}_{t+1})
\end{equation}
For a set of paths $\textbf{p}$ which start from $\hat{s}$, one may compute the probability of observing any path within this set by summing over the probability of all paths within $\textbf{p}$ as,
\begin{equation}
\label{eq:path_set_prob}
    \mathbb{P}_{\sysmdp,\pi}^{\hat{s}}[\textbf{p}] = \sum_{p\in \textbf{p}} \mathbb{P}_{\sysmdp,\pi}^{\hat{s}}(p)
\end{equation}
\paragraph{Safety Property.} We are interested in finite and infinite horizon probabilistic safety properties. Given a set of unsafe system states $S_U \subseteq S$, the set of unsafe states in MDP $\sysmdp$ is given by $\hat{S}_U := \{\hat{s} \in \hat{S}~|\forall s\in S_U. ~\hat{s}=(s,\_,\_)\}$ and the set of paths that reach unsafe states within $n$ steps from a state $\hat{s}$ is $\unsafe_{\hat{S}_U,n}(\hat{s}) := \{(\hat{s}_t)_{t=0}^{\infty} \in \paths_{\sysmdp,\kappa}(\hat{s})~|~\exists t\leq n.~\hat{s}_t \in \hat{S}_U\}$. The system is \emph{safe} with respect to a finite horizon of length $n$ and probability $p$ if, 
\begin{equation}
\label{eqn:syssafe}
\max_{\pi \in \hat{S}\rightarrow A} \mathbb{P}_{\sysmdp,\pi}^{\hat{\iota}}[\unsafe_{\hat{S}_U,n}(\hat{\iota})] \leq p
\end{equation}
A safety property $\psi$ of this form can be expressed in Probabilistic Computational Tree Logic (PCTL) as the property $\mathbb{P}_{\leq p} [F^{\leq n} \hat{S}_U]$ where $\mathbb{P}$ is PCTL's probabilistic operator, $F$ is the eventually modality, and, abusing notation, $\hat{S}_U$ is the proposition that the MDP state is in the set $\hat{S}_U$. We write $\sysmdp,\hat{s} \models \phi$ to refer to the fact that the set of paths of the MDP $\sysmdp$ starting from state $\hat{s}$ satisfy the property $\psi$. If the MDP $\sysmdp$ satisfies Eqn.~\ref{eqn:syssafe}, then $\sysmdp,\hat{\iota} \models \psi$.

\subsection{Perfect-Perception MDP with Shield}
\label{sec:PP_MDP_Shielded}

An autonomous system is a \emph{perfect-perception} system if 
\begin{equation}
\forall s\in S.~\forall o \in \{o'\in O~|~\omega(s,o')>0\}.~\rho(o)=s 
\end{equation}
Intuitively, this says that the agent is always able to perceive the underlying actual state from the observations.
Under this condition, we can define the MDP  corresponding to the perfect-perception variant of the autonomous system as $\ppsysmdp := (S,A,\iota, P)$ where $P:S\times A \times S\rightarrow [0,1]$ is defined as:
\begin{equation}
\label{eqn:perf_p_def}
\begin{array}{ll}
P := & \forall s,s' \in S. \\
& (s) \xrightarrow[]{a \in A,~\eta(s,a,s')} (s')
\end{array}
\end{equation}
Note that this MDP does not need to track the observations and state estimates due to perfect perception.
Given the perfect-perception MDP $\ppsysmdp$ and a PCTL safety property $\psi$ of the form $F^{\leq n}\hat{S}_U$, a \emph{shielded} perfect-perception MDP $\shieldppsysmdp$ is the tuple $(S,A,\iota, P)$ where $P:S\times A \times S\rightarrow [0,1]$ is defined as:

\begin{equation}
\begin{array}{ll}
P := & \forall s,s' \in S. \\
& (s) \xrightarrow[]{a \in \Delta_{\psi,\lambda}(s),~\eta(s,a,s')} (s')
\end{array}
\end{equation}
where $\Delta_{\psi,\lambda}: S \rightarrow 2^A$ is referred to as the shield and it restricts the actions available in the MDP in a state-dependent manner. In this work, we only consider \emph{absolute} shields  constructed as follows:
\begin{align}
\label{eqn:shield}
\Delta_{\psi,\lambda}(s) &:= \{a\in A~|~\sigma_\psi(s,a) \leq \lambda\}\\
\label{eqn:sigma_def}
\text{where } \sigma_\psi(s,a) &:= \min_{\pi\in S\rightarrow A} \sum_{s'\in S} P(s,a,s')\cdot\mathbb{P}_{\ppsysmdp,\pi}^{s'}[\unsafe_{S_U,n}(s')]
\end{align}

\begin{theorem}[Global Safety of Perfect-Perception MDP with Shield]
\label{thm:global}
Given a shielded perfect-perception MDP $\shieldppsysmdp$ where the shield is constructed with respect to a safety property $\psi:=F^{\leq n}S_U$, starting from the initial state $\iota$ of $\shieldppsysmdp$, and under the assumptions that $\iota$ is safe (\textbf{initially safe}) and the system does not encounter stuck states (\textbf{no stuck}), the maximum probability of reaching unsafe states $S_U$ in $n'$ steps is given by,
\begin{align}
\max_{\pi \in S\rightarrow A} \mathbb{P}_{\Gamma_{\mathrm{Perf}}^\Delta,\pi}^{\iota}[\unsafe_{S_U,n'}(\iota)] \leq 1-(1-\lambda)^{n'}
\end{align}
\end{theorem}
\begin{proof}
The following provides a summary of a complete proof found in Appendix~\ref{sec:proof_thm1}. This proof relies on the assumptions \textit{initially safe} and \textit{no stuck}, which are referenced in Theorem~\ref{thm:global} and precisely defined in Appendix~\ref{sec:proof_thm1}. For ease of construction, we prove a lower bound on safety using an equivalent statement of Theorem \ref{thm:global}: $\min_{\pi \in S\rightarrow A} \mathbb{P}_{\shieldppsysmdp,\pi}^{\iota}[\safe_{S_U,n'}(\iota)] \ge (1-\lambda)^{n'}$.
This proof uses \textit{occupancy vectors} of type $\mathbb{R}^S$ which represent the probability the system is in each state at a given time. First, we define a $\psi$-preserving transformation of a system $\Gamma_\text{Perf}$, which we call $\Gamma_\text{Perf}^\psi$, which allows one to measure the probability of achieving $\psi$ on $\Gamma_\text{Perf}$ by simply considering the probability of $\Gamma_\text{Perf}^\psi$ ending in $S_U$. Defining this transformation allows us to reason about the satisfaction of $\psi$ in $\Gamma_\text{Perf}$ by only considering individual time-steps of $\Gamma_\text{Perf}^\psi$ rather than entire paths on $\Gamma_\text{Perf}$. Then, we build an inductive argument over $d_i$, the occupancy vector achieved by $\Gamma_\text{Perf}^\psi$ at timestep $i$, and a particular function $g : \mathbb{R}^S \rightarrow \mathbb{R}$. By showing that any policy $\pi$ allowed by the shield implies $g(d_{i+1}) \ge (1-\lambda)g(d_i)$ and assuming $g(d_0)=1$, we conclude $g(d_{n'})\ge (1-\lambda)^{n'}$ via an inductive argument. We use the property that $g(d_{n'})$ is a lower bound on the safety of $d_{n'}$ to conclude the probability of $\Gamma^\psi_\text{Perf}$ ending in a safe state must be greater than $(1-\lambda)^{n'}$ after $n'$ time-steps. Using the correspondence of $\Gamma^\psi_\text{Perf}$ and $\Gamma_\text{Perf}$ described above, this implies the restatement of Theorem \ref{thm:global}.
\qed
\end{proof}

It is important to note the distinction between $n$ and $n'$ 
in Theorem \ref{thm:global}. The parameter $n$ represents the lookahead used to construct the shield $\Delta$, while the parameter $n'$ represents the horizon of the system considered for safety verification. It is somewhat counterintuitive that $n$ does not appear in the safety condition provided by Theorem \ref{thm:global}. Figure \ref{fig:worst_case_safety} and its discussion show why increasing lookahead does not improve the safety guarantee of shielded control on an arbitrary system.  However, increasing lookahead may improve safety for a fixed system through mechanisms not captured by Theorem 1.

We note that \emph{stuck} states are those states where the shield is empty and therefore, no further actions are allowed by the agent. Such stuck states are distinct from safe and unsafe states. As we show in Appendix~\ref{sec:entrap},  if the shield is synthesized with respect to a fixed horizon, there exist certain state topologies for which actions allowed by the shield are guaranteed to result in stuck states. In practice, for systems which can abort operation or revert to manual control, it is reasonable to assume that autonomous control will be terminated when a state is reached where no actions are allowed by the shield. Alternately, detecting that a system is in a stuck state can be viewed as an \emph{early-warning} indicating that further use of autonomous control is fraught with risks. A pre-deployment analysis of the frequency of reaching stuck states can also serve as useful feedback to system designers. %
Note that, for a particular system, the likelihood of reaching stuck states can be changed by increasing the shielding lookahead $n$ or modifying the shielding parameter $\lambda$.\footnote{In the case study found in Figure \ref{fig:scarbros_demon}, either of these modifications are sufficient to make stuck states unreachable by a shielded controller.}

Theorem 1 represents a significant advancement of the theoretical foundations of shielded control. Previous methods \cite{Fulton_Platzer_2018,jansen2019safe} shield actions according to their safety relative to the safest action from the same state. This strategy does not produce a global guarantee of safety because even the safest action from a state may be entirely unsafe. Instead, we apply the shielding criterion uniformly across all actions to produce a global guarantee of safety, but this comes at the cost of possible stuck states with no available actions allowed by the shield. In the absence of these stuck states, Theorem 1 gives a global guarantee of safety relative to the parameters of the shield. %

\subsection{Imperfect-Perception MDP with Conformal Shield}
\label{sec:IP_MDP_Shielded}
The shielded, conformalized version of the system  with imperfect perception is the tuple $(S,O,A,\hat{\iota},\hat{\rho}_\alpha,\eta,\omega, \hat{\Delta}_{\psi,\lambda})$ where $\hat{\rho}_\alpha:O\rightarrow 2^S$ is the conformalized perception function and $\hat{\Delta}_{\psi,\lambda}:2^S\rightarrow 2^A$ is a shield that acts on a set of states.
This system can be formulated as a Markov Decision Process (MDP), $\shieldcpsysmdp := (\hat{S},A,\hat{\iota}, P)$, where $\hat{S} := S \times O \times 2^S$ is the set of MDP states, $\hat{\iota} := \iota \times \bot_O \times \{\iota\}$ is the initial MDP state, and $P:\hat{S}\times A \times \hat{S}\rightarrow [0,1]$ is the (partial) probabilistic transition function defined as:
\begin{equation}
\begin{array}{ll}
P := & \forall s,s'' \in S,\forall \overline{s} \in 2^S, \forall o,o' \in O. \\
& (s,m,\overline{s}) \xrightarrow[]{a \in \hat{\Delta}_{\psi,\lambda}(\overline{s}),~\eta(s,a',s'')\cdot\omega(s'',o')} (s'',o',\hat{\rho}(o'))
\end{array}
\end{equation}
The shield $\hat{\Delta}_{\psi,\lambda}$ is constructed from the shield $\Delta_{\psi,\lambda}$ for the perfect-perception MDP $\ppsysmdp$ as:
\begin{equation}
\label{eqn:lifted_shield}
\hat{\Delta}_{\psi,\lambda}(\overline{s}) := \bigwedge_{s\in\overline{s}} \Delta_{\psi,\lambda}(s)
\end{equation}
\paragraph{\textbf{Conformalized Perception.}} To construct the conformalized perception function $\hat{\rho}_\alpha$, we use existing methods for conformalizing classifiers~\cite{angelopoulos2021gentle}. We assume access to a \emph{calibration dataset} consisting of i.i.d. pairs $(o_i,s_i)$ of observations $o_i \in O$ and corresponding actual system states $s_i \in S$, unseen during training of the perception function $\pi:O\rightarrow S$. 
The conformalized perception function $\hat{\rho}_\alpha$, constructed using this data, comes with the following probabilistic guarantee:
\begin{equation}
\label{eq:conf_guar}
    \mathbb{P}[s_\text{test}\in \hat{\rho}_\alpha(o_\text{test})] \geq 1 - \alpha
\end{equation}
where $(o_\text{test},s_\text{test})$ is a fresh observation-state pair from the same distribution as the calibration dataset.

Ideally, we would like to provide a \emph{local safety} guarantee for imperfect-perception MDPs with conformal shields. Given an imperfect-perception MDP $\shieldcpsysmdp$ with conformal shield $\hat{\Delta}_{\psi,\lambda}$ where the shield is constructed with respect to a safety property $\psi:=F^{\leq n}\hat{S}_U$, we would like to guarantee that
\begin{equation}
\label{eqn:local_safe}
    \mathbb{P}[\sigma_{\psi}(\hat{s},a) \leq \lambda] \geq 1 - \alpha
\end{equation}
\noindent where $a\in\hat{\Delta}_{\psi,\lambda}(\hat{s})$, $\sigma_\psi(\hat{s},a) := \min_{\pi\in S \rightarrow A} \sum_{s'\in S} P(s,a,s')\cdot\mathbb{P}_{\ppsysmdp,\pi}^{s'}[\unsafe_{S_U,n}(s')]$, $\hat{s} := (s,m,\overline{s})$, and the probability measure in Equation~\ref{eqn:local_safe} corresponds to the distribution over $(o_i,s_i)$ pairs induced by the distribution $\mathbb{P}_{\shieldcpsysmdp,\pi}^{\hat{\iota}}$ over paths of the MDP $\shieldcpsysmdp$ where $\pi$ is the worst-case policy for safety.

Intuitively, this would guarantee that after applying an action prescribed by the shield in the current state $\hat{s}$, there exists a policy from the resulting successor states such that, assuming perfect perception in the future, the statement that the probability of reaching unsafe in the future is bounded holds with a high probability. While such a guarantee assumes a best-case scenario in the future, it can be helpful to know that the chosen action is at least safe in an optimistic setting.

Establishing such a guarantee would be straightforward if the calibration dataset were drawn from the same distribution over $(o_i,s_i)$ pairs as in Equation~\ref{eqn:local_safe}. In practice, however, the calibration dataset is initially drawn from the distribution over $(o_i,s_i)$ pairs induced by the path distribution for the imperfect-perception MDP $\sysmdp$ that is neither shielded nor conformalized and where the actions in each step are drawn uniformly at random. One can then use the following iterative procedure to improve the calibration dataset: (i) use the initial calibration dataset to construct the conformalized perception function; (ii) construct a shielded, conformalized MDP and use this MDP (via simulations) to collect more calibration data; (iii) use the new calibration dataset to construct a new conformalized perception function, and repeat. We leave the investigation of such a procedure and its ability to help establish the desired local safety guarantee as future work. Ideally, we would also like to provide a global safety guarantee, similar to Theorem~\ref{thm:global}, that bounds the probability of reaching unsafe states as long as the agent takes actions prescribed by the shield. We leave this also for future work.

\subsection{Probabilistic Abstractions of Conformalized Perception}
\label{sec:Abs_MDP_Shielded}
The observation function $\omega:S\times O \rightarrow [0,1]$ is very complex in practice and infeasible to model mathematically. To be able to model check $\shieldcpsysmdp$, we leverage past work~\cite{calinescu2023IMRPSV,PasareanuMGYICY23} to construct a probabilistic abstraction of the composition of the conformalized perception function $\hat{\rho}$ and $\omega$. The resulting  shielded, conformalized system  with imperfect perception is the tuple $(S,A,\hat{\iota},\nu,\eta,\hat{\Delta}_{\psi,\lambda})$ where $\nu:S\times 2^S \rightarrow [0,1]$ is the probabilistic abstraction that describes a distribution over sets of estimated system states given an actual system state. 
This system can be formulated as a Markov Decision Process (MDP), $\shieldcpprsysmdp := (\hat{S},A,\hat{\iota}, P)$, where $\hat{S} := S \times 2^S$ is the set of MDP states, $\hat{\iota} := \iota \times \{\iota\}$ is the initial MDP state, and $P:\hat{S}\times A \times \hat{S}\rightarrow [0,1]$ is the (partial) probabilistic transition function defined as:
\begin{equation}
\begin{array}{ll}
P := & \forall s,s' \in S,\forall \overline{s} \in 2^S. \\
& (s,\overline{s}) \xrightarrow[]{a \in \hat{\Delta}_{\psi,\lambda}(\overline{s}),~\eta(s,a,s')\cdot\nu(s',\overline{s}')} (s',\overline{s}')
\end{array}
\end{equation}

\section{Empirical Evaluation}
\label{sec:eval}

We evaluate our proposed approach on a simulated case study that mimics the autonomous airplane taxiing system from prior works~\cite{KadronGPY21,fremont2020formal,PasareanuMGYICY23}. Note that all of our empirical evaluation is with respect to imperfect-perception systems for which Theorem~\ref{thm:global} is not applicable.

\subsection{Implementation}

Figure~\ref{fig:pipeline} describes our overall pipeline.\footnote{Code available at \url{https://github.com/CSU-TrustLab/cp-control}.} 
It is implemented in Python along with external tools such as the PRISM model checker~\cite{PRISM}. The \emph{conformalizer} implements a standard conformal prediction algorithm (Figure 2 from \cite{angelopoulos2021gentle}) and yields the conformalized perception DNN from the original perception DNN of the agent. The \emph{shield synthesizer} is implemented using PRISM. In particular, for each state $s \in S$ and each action $a \in A$ of the perfect-perception MDP, we issue a PRISM query as described in Eqn.~\ref{eqn:sigma_def} and construct the shield using Eqn.~\ref{eqn:shield}. The \emph{DNN evaluator} simply evaluates the conformalized DNN on the test data to construct a confusion matrix that records, for each actual state, the number of times each set of estimated states is predicted. 
We also implement a simple \emph{compiler} in Python that converts perfect-perception MDPs expressed in a domain specific language to imperfect-perception MDPs with a conformal shield and a probabilistic abstraction, i.e., $\shieldcpprsysmdp$ (as in Section~\ref{sec:Abs_MDP_Shielded}) expressed in the PRISM modeling language. The compiler ingests the shield computed by the shield synthesizer and the confusion matrix for this translation. 
Finally, we use PRISM as our \emph{model checker} to model check the resulting MDP.

\subsection{Details of Autonomous Taxiing Case Study}

For our case study, we consider an experimental system for autonomous centerline tracking on airport taxiways~\cite{KadronGPY21,fremont2020formal}. It takes
as input a picture of the taxiway and estimates the plane’s position with respect to the
centerline in terms of two outputs: cross-track error (cte), which is the distance in
meters of the plane from the centerline and heading error (he), which is the angle
of the plane with respect to the centerline. The state estimation or perception is performed by a neural network. The resulting state estimates are
fed to a controller, which maneuvers the plane to follow the centerline. Error is defined
as excessive deviation from the centerline, either in terms of cte or he values. We consider a discretized version of the system as in prior work~\cite{PasareanuMGYICY23}.

\noindent{\emph{\textbf{Simulation Setup.}}}
To explore with realistic data, we devised a simulated setup which has the same objective as the autonomous airplane taxiing problem. Specifically, we implemented a simulated setup of a Turtlebot4 robot in the Gazebo simulator operating on a taxiway. In our simulation, the Turtlebot4 serves the role of the airplane and it is equipped with an RGB camera which can capture $240 \times 320$ pixel images. We use the simulator to collect a dataset of images that would be captured by the camera. This collected dataset is used for training the perception DNN and also for conformalizing the DNN.

\begin{figure}[t]
    \centering
    \includegraphics[width=0.5\linewidth]{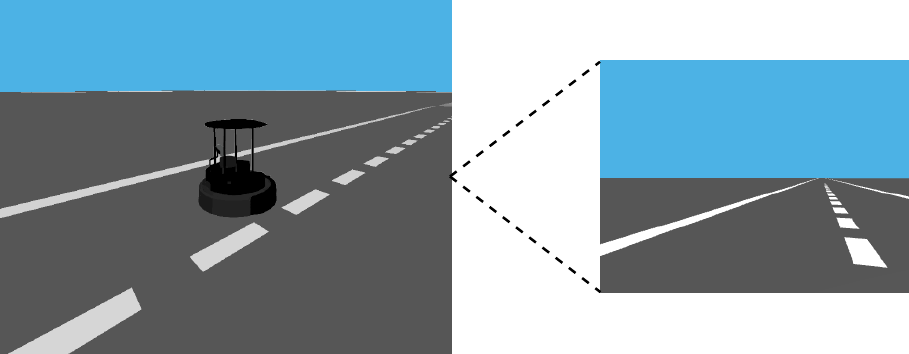}
    \caption{The simulation setup of the Turtlebot4 in Gazebo (left), and a sample image captured by the Turtlebot4's camera (right). 
    These captured images represent the images that would be passed to the perception DNN.}
    \label{fig:simulator_setup}
\end{figure}

\noindent{\emph{\textbf{Data Processing Pipeline.}}}
To generate a dataset of images, we place the Turtlebot4 randomly on the taxiway, and save the resulting image taken from the camera. Along with this image, we capture the true pose of the robot, $(x, y, \theta)$, where $(x,y)$ are the robot's coordinates and $\theta$ is the angle with respect to the centerline. We collected a total of 5500 images. Fig.~\ref{fig:simulator_setup} contains screenshots of the simulated setup and an example image taken from the camera. Input preprocessing normalizes RGB values using ImageNet statistics ($\mu=[0.485, 0.456, \\ 0.406]$, $\sigma=[0.229, 0.224, 0.225]$) to leverage pre-trained backbone weights for perception DNNs. For training the DNN, we further augment the dataset using strategies such as random horizontal flipping ($p=0.5$), color jittering with brightness and contrast variations (±0.2), and Gaussian noise injection ($\sigma=0.01$) to improve model robustness and generalization.
We use 60\% of the collected data for training the DNN, 10\% for validation, 15\% for calibration, and 15\% for testing. 
The dataset is split using stratified sampling to maintain balanced class distributions across all cte and he categories. This ensures representative learning across all states and prevents bias toward more frequent state configurations.

\noindent{\emph{\textbf{State Space Representation.}}} As in past work
~\cite{PasareanuMGYICY23}, we consider a discretized state space where the safe range of cross-track error (cte) is discretized into 5 buckets
and safe range of heading error (he) into 3 buckets.
This yields a total of 15 discrete system states that perception must classify from visual input.

\noindent{\emph{\textbf{Neural Architecture.}}} The perception function $\rho$ is implemented as a multi-task convolutional neural network processing RGB images of size $240 \times 320 \times 3$. The architecture employs a MobileNetV2 backbone utilizing depthwise separable convolutions with inverted residual blocks for computational efficiency. Feature extraction layers progressively reduce spatial dimensions from $240 \times 320$ to $7 \times 10$ while expanding the channel depth to 1280.
The global average pooling aggregates spatial features into a unified 1280-dimensional representation. The shared processing backbone applies dropout regularization ($p=0.3$), batch normalization, and ReLU activation, with linear transformations that reduce dimensionality through $1280 \to 512 \to 256$ neurons.

The architecture then branches into task-specific classification heads from the 256-dimensional shared representation. The cte head implements a $(256 \to 128 \to 5)$ structure with ReLU activations, computing probability distributions over the five lateral position states. The he head follows an identical architecture $(256 \to 128 \to 3)$ for the three heading states. Both heads employ softmax activation to produce normalized probability distributions over their respective state spaces.

The trained model outputs joint probability distributions $P(\text{cte}|o)$ and $P(\text{he}|o)$ for each input observation $o$, providing both point predictions through argmax selection and uncertainty estimates through probability values. These probabilistic outputs $\hat{f}$ serve as input to the conformalization process.

\noindent{\emph{\textbf{Neural Training.}}} The network jointly optimizes both classification tasks 
with independent cross-entropy losses: $\mathcal{L}_{\text{cte}} = -\sum_{i} y_i^{\text{cte}} \log(p_i^{\text{cte}})$ and $\mathcal{L}_{\text{he}} = -\sum_{i} y_i^{\text{he}} \log(p_i^{\text{he}})$, where $y_i$ represents one-hot encoded ground truth labels and $p_i$ denotes predicted probabilities. The combined objective $\mathcal{L}_{\text{total}} = \mathcal{L}_{\text{cte}} + \mathcal{L}_{\text{he}}$ enables the simultaneous learning of both navigation parameters.

For training, we use Adam optimization with $\beta_1=0.9$, $\beta_2=0.999$ which provides adaptive learning rate control with initial rate $0.001$. 
The learning rate scheduler monitors the validation loss with patience=5 epochs, applying multiplicative decay (factor=0.1) when improvement stagnates. Regularization of weight decay ($\text{factor}=1 \times 10^{-4}$) 
prevents overfitting, while gradient clipping ($\text{max\_norm} = 1.0$) ensures training stability. Mini-batch processing with $\text{batch\_size}=16$ balances memory efficiency with gradient estimation quality.

\noindent{\emph{\textbf{MDP Formulation.}}}
We formulate a perfect-perception MDP corresponding to our case study where the state space $S:=\{0,\ldots,4\}\times\{0,1,2\}$ representing the five possible cte values and three possible he values and the action space $A:=\{0,1,2\}$ corresponds to the three actions available for each state, namely, go straight, go left, and go right. In each time-step, the agent continues to move forward even if it executes the command to go left or go right. We define a stochastic dynamics for the agent where the source of stochasticity is in the action execution. For instance, the agent controller might choose to go straight in a particular state, but we model a small probability of this choice actually resulting in the agent going left or right due to windy conditions or actuator errors. For all our experiments, we assume that the agent starts at the centerline, facing straight ahead (i.e., in the initial state, cte=0 and he=0). A PRISM model of the perfect-perception MDP is available in Appendix~\ref{sec:prism_model}.
Note that while the PRISM model makes some changes for ease of expression, in essence it models a perfect-perception MDP as explained here.

The autonomous taxiing agent fails the mission if it moves off the taxiway, or if the agent rotates too far left or right. These states will be referred to as \emph{fail} state and are terminal states in the model. Moreover, the imperfect-perception agent with a conformal shield may find itself in a state where the set of predicted states (output by the conformalized perception) are such that the shield (as constructed using Eqn.~\ref{eqn:lifted_shield}) has no safe actions available for the predicted set. Such states are referred to as \emph{stuck} states and are also terminal states in our MDP.\footnote{An empty shield is a commonality between this definition of \textit{stuck} and the one given in Section \ref{sec:PP_MDP_Shielded} in the perfect-perception context. However, the conditions which produce stuck states in the two cases are different because of the distinct shield constructions.
} 
In practice, once the agent encounters a stuck state, it may employ a failsafe action (take another reading, go straight, halt the agent, raise an alarm to a human operator, etc.) but this is outside the scope of our current work.

For our experiments, we construct the imperfect-perception MDP with conformal shielding and probabilistic abstractions (as defined in Section~\ref{sec:Abs_MDP_Shielded}) using our compiler. The perfect-perception shield is constructed with respect to the PCTL property, $F^{\leq 5}S_U$ where $S_U$ are the fail states for the case study. We construct multiple MDPs by varying the conformalization hyperparameter $\alpha$ ($(1-\alpha) = \alpha' \in \{0.95, 0.99, 0.995\}$), and the shield hyperparameter $\lambda$ ($(1 - \lambda)= \lambda' \in \{0.7, 0.8, 0.9\}$). With each of the imperfect-perception MDP, we employ three different versions:
\begin{enumerate}
    \item An MDP that has non-determinism intact for selecting actions (referred to as \emph{worst-case})
    \item An MDP where we fix the policy to choose actions allowed by the shield uniformly at random (referred to as \emph{random})
    \item An MDP where where we fix the policy to choose the safest action as per the shield; given a set of predicted states $\overline{s}$, the safest action $a_{safe}$ is the action with the smallest maximum probability of unsafety across the set of states, i.e., $a_{safe} := \argmin_{a\in \hat{\Delta}_{\psi,\lambda}(\overline{s})} \max_{s\in\overline{s}} \sigma_\psi(s,a)$ where $\sigma_\psi(s,a)$ is as defined in Eqn.~\ref{eqn:sigma_def} (referred to as \emph{safest})
\end{enumerate}

We use the PRISM model checker to analyse the models constructed with respect to the following PCTL properties:
\begin{equation}
    \begin{split}
        \mathbf{P_{\textsf{fail}}}&:\: \mathbb{P}_?\: [F^{\leq n'}\: \text{``fail''}] \\
        \mathbf{P_{\textsf{stuck}}}&:\: \mathbb{P}_?\: [F^{\leq n'}\: \text{``stuck''}] \\
        \mathbf{P_{\textsf{success}}}&:\: \mathbb{P}_?\: [F\: t=n']
    \end{split}
\end{equation}
where $n'$ is the length of the agent paths (i.e., number of time-steps) and $t$ is a variable denoting the number of steps taken by the agent. $\mathbf{P_{\textsf{success}}}$ captures the situation where the agent successfully operates for $n'$ steps without failing or getting stuck. In our experiments, we consider values of $n'$ in $\{1,\ldots,30\}$. For the MDP models using non-determinism for action selection (\emph{worst-case}), the properties were slightly modified; $\mathbf{P_{\textsf{fail}}}$ and $\mathbf{P_{\textsf{stuck}}}$ are aimed at finding the strategy which maximizes these values, and $\mathbf{P_{\textsf{success}}}$ is aimed at finding the minimal value. This is to find the worst case scenario, i.e. where the system chooses the worst possible action allowed by the shield at each step.

As a baseline, we consider an imperfect-perception MDP formulation of the autonomous taxiing agent where the perception DNN is not conformalized and the shield is constructed assuming perfect perception.

\subsection{Empirical Results}

As the length of the operation increases, the probabilities of encountering either the fail state or the stuck state increases, therefore a decrease of mission success (see Fig.~\ref{fig:safety_results}). 
The conformaliser hyperparameter $\alpha'$ (i.e., $1-\alpha$)  has a substantial effect on the performance of the agent. As expected, higher $\alpha'$ decreases the probability of going to the fail state.
Increasing $\alpha'$, however, results in the conformalized perception predicting larger sets of states. This in turn increases the probability of the shield being empty, i.e., no safe action being available and the agent getting stuck as reflected in the results. Note that, in our baseline setup, for all states, there was a minimum of one action allowed by the shield, meaning the system would never get stuck.

The shield hyperparameter $\lambda'$ (i.e., $1-\lambda$) also affects the agent significantly. As $\lambda'$ increases, fewer actions are included in the shield. This therefore increases the number of sets of states where the shield is empty, resulting in an increase of probability of encountering the stuck state as seen in the results.
A consequence of getting stuck often is that the probability of encountering the fail state as well as succeeding in the mission are reduced.
That said, when the conformalizer guarantee is set to 0.95, there is only a small increase to the probability of encountering the stuck state. It is then seen that the probability of going to a fail state decreases as the shield value increases, demonstrating the shield is further filtering unsafe actions. 

We also compare our results with a baseline imperfect-perception MDP model of the agent where the perception is not conformalized and the shield is constructed assuming perfect perception.
The baseline results demonstrate that the inclusion of conformalization generally reduces the probability of crashing ($\mathbf{P_{\textsf{fail}}}$). While the baseline has a higher probability of succeeding ($\mathbf{P_{\textsf{success}}}$) in some instances, this is due to the impossibility of getting stuck. For instance, when $\alpha'=0.95$, the conformalized MDP demonstrates a smaller probability of getting stuck and outperforms  baseline with respect to $\mathbf{P_{\textsf{success}}}$. In contrast, for higher values of $\alpha'$, probability of success is lower compared to the baseline because the agent is very likely to get stuck; for the same reason, the conformalized agent  demonstrates a lower probability of failure compared to the baseline.

\begin{figure}[t]
    \centering
    \includegraphics[width=1.0\linewidth]{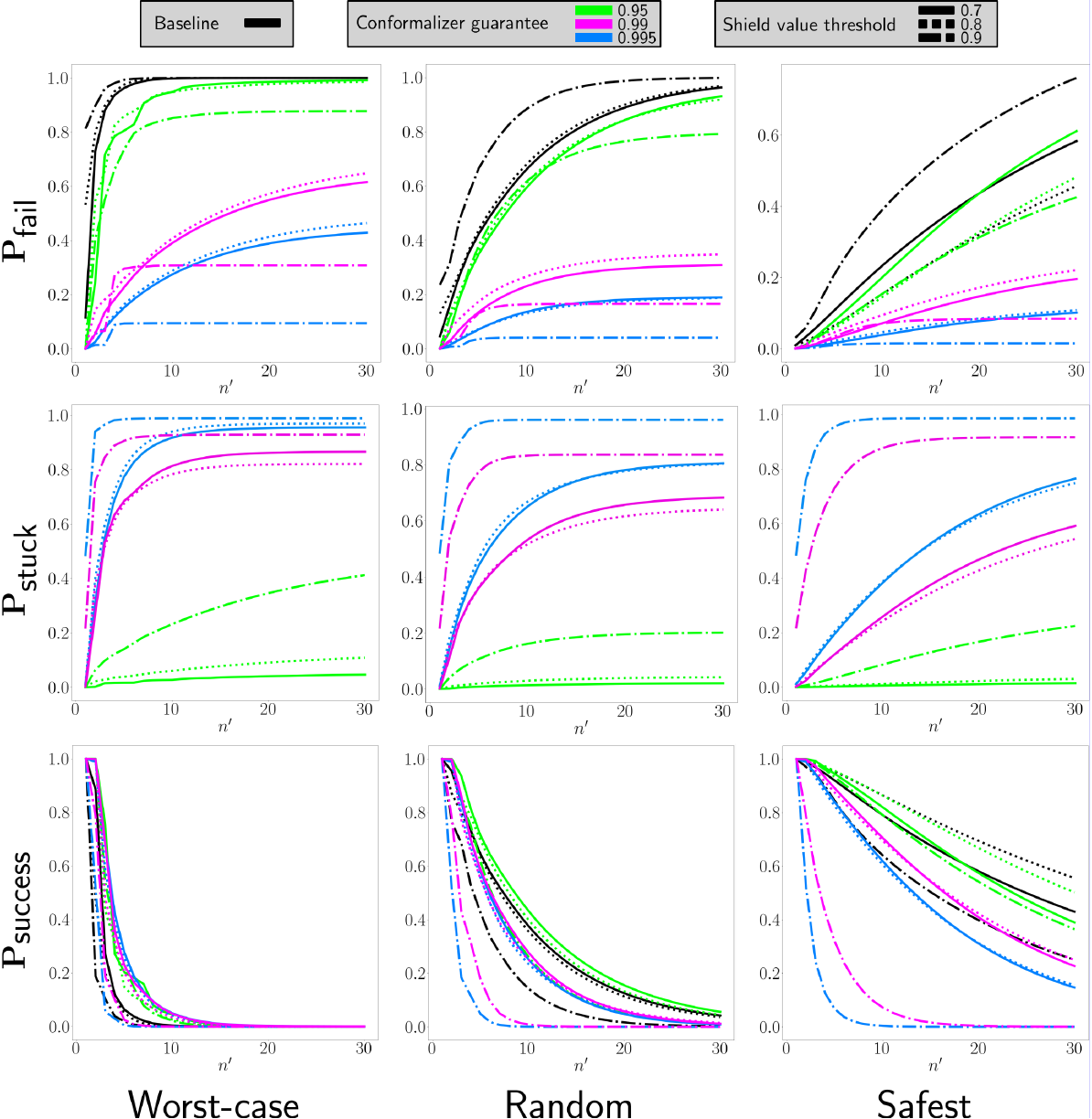}
    \caption{Results for properties $\mathbf{P_{\textsf{fail}}}$, $\mathbf{P_{\textsf{stuck}}}$, and $\mathbf{P_{\textsf{success}}}$ as $n'$ increases. Black lines represent the baseline results.}
    \label{fig:safety_results}
    \vspace{-0.2in}
\end{figure}

We also evaluate the time to perform model checking for each property across $n'$, see Fig.~\ref{fig:time_checking_results}. Expectedly, as $n'$ increases the model checking time also increases, with $n'=1$ on average being less than 0.01 seconds, and $n'=30$ ranging between 0.2 and 1.5 seconds depending upon property and setup. As $\alpha'$ increases, more states are included in the sets, increasing the overall model size. Though, continuing to increase $\alpha'$ will inevitably result in fewer unique sets being produced as a result of the sets including more and more of the states. Therefore, this will result in a smaller model for analysis purposes. This is observed as $\alpha'$ increases from 0.95 to 0.99 with the accompanying increase in time, and then a decrease in time when $\alpha'$ increases from 0.99 to 0.995. Increasing $\lambda'$ can reduce the number of available actions for each set of states and can lead to a reduction in model checking time as the number of possible transitions decreases, specially if the shield value threshold generates substantially more empty action sets for the predicted sets of states. 

\begin{figure}
    \centering
    \includegraphics[width=1.0\linewidth]{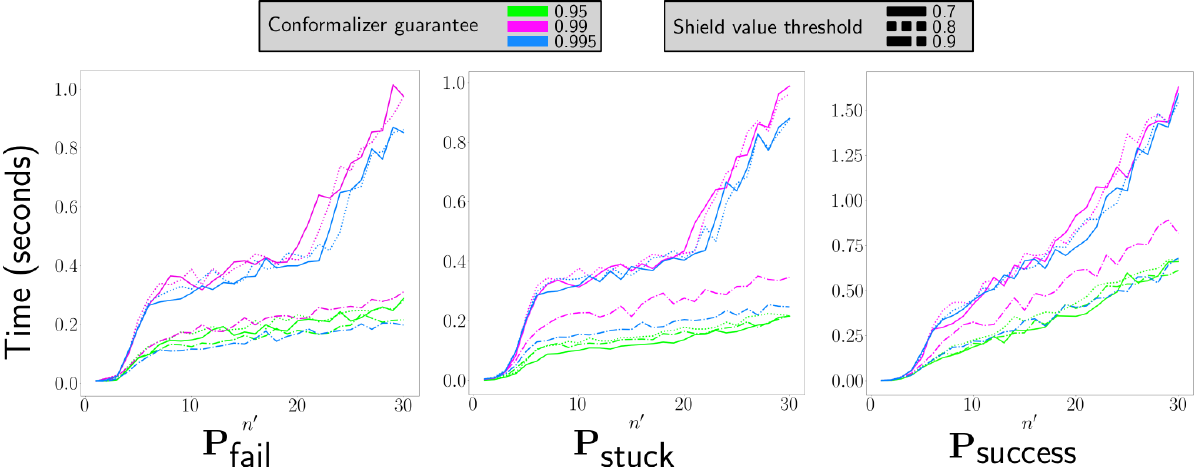}
    \caption{Time to perform model checking on each property as $n'$ increases. For each variation of $\alpha'$ and $\lambda'$, the average time was computed over the three different model variations---worst-case, random, and safest.}
    \label{fig:time_checking_results}
\end{figure}

\subsection{Discussion}
We note that our proposed shield can be used either in an offline or online fashion. Similar to existing works on shielding~\cite{jansen2019safe,carrNJT2023safe}, the shielded imperfect-perception MDP can be used when learning an optimal policy via reinforcement learning (RL) to ensure safe exploration while accelerating learning. Once an optimal policy is learned, the shield may be discarded. In contrast to this offline usage, the shield may also be used online as a run-time monitor that ensures no actions are taken by the agent that violate local safety. This mode of usage is particularly useful when the optimal policy does not come with any safety guarantees. The synthesis of the shield itself can also be carried out in an offline~\cite{jansen2019safe} or online~\cite{konighofer2023online} fashion as has been noted in prior works. Since our conformal shield maps sets of predicted states to sets of actions, the size of the shield grows exponentially with the number of possible classes that the original perception predicts over. Consequently, online shield construction can be particularly effective since the shield is only computed for the sets of predicted states observed in operation. We also note that, as in existing works on shielding, our construction also assumes that the safety-relevant fragment of the MDP can be mathematically modeled. However, our work does not require modeling the observation or perception processes. Finally, we already note in Section~\ref{sec:IP_MDP_Shielded} that the conformalization guarantee requires distributional assumptions and careful collection of calibration data.

\section{Related Work}
\label{sec:related}
\paragraph{\textbf{Shielding.}} %

The work on shielding Partially Observable Markov Decision Processes (POMDPs)~\cite{carrNJT2023safe,sheng2024safe} is closely related to our approach.
Like MDPs, POMDPs model decision-making settings that involve both stochasticity and non-determinism. However, in contrast to MDPs, POMDPs also capture the fact that the underlying state can only be accessed through stochastic observations. Safety is ensured by constructing a shield that restricts the set of available actions in each belief state of the POMDP.\footnote{A \emph{belief state} is a distribution over the underlying states that reflects the agent's current best estimate based on the history of observations.}

While promising, POMDP-based shielding faces several limitations. First, it requires modeling the stochastic process of generating observations from underlying states. For the airplane taxiing example, one would need to model the complex process of generating images corresponding to the airplane’s position on the taxiway, which can be infeasible in practice.\footnote{\cite{carrNJT2023safe} slightly relax this requirement by only needing to know the transitions from states to observations that have non-zero probability, rather than the exact probabilities.} Second, a POMDP formulation assumes that agents compute belief states before choosing actions. However, this assumption often does not hold for realistic systems. Third, existing works on POMDP shielding~\cite{carrNJT2023safe,sheng2024safe} only consider infinite-horizon, almost-sure safety properties (i.e., the probability of violation tends to zero as the run length increases). However, for many agents, including the autonomous taxiing system, the probability of violating safety over an infinite horizon is often one. Therefore, we are interested in finite-horizon, probabilistic safety guarantees.

On the other hand, \cite{Goodall2024B} address the challenge of achieving safe reinforcement learning (RL) in complex, continuous environments. They implement a counterexample-guided policy optimization based on reward computation. However, this study only provides approximate safety guarantees. \cite{Konighofer2020LJB} explore shielding for RL to synthesize a correct-by-construction reactive system and further extend this with a formal “probabilistic shield” approach. This method computes per-state/action safety constraints via shielding, filtering the RL agent's actions focusing on the tradeoffs
between sufficient progress in exploration of the environment
and ensuring safety%
~\cite{jansen2019safe}. Nevertheless, they do not formally guarantee a global safety of existing shield constructions for perfect-
perception.
\paragraph{\textbf{Conformal Prediction and Control.}}
Conformal prediction offers a powerful framework for generating finite-sample, distribution-free error bounds for perception, which can then be integrated into robust control pipelines. End-to-end probabilistic safety guarantees for perception-based control with sequential, time-uniform conformal bounds have been explored by Yang et al. \cite{yang2023safe} using barrier functions. Building on this, Lindemann et al. \cite{lindemann2024formal} effectively represent uncertainty via conformal prediction and formalize safety through linear programming. These approaches focus on continuous time and therefore provide weaker formal guarantees.
State-dependent (heteroskedastic) conformal bounds have been introduced by Waite et al. \cite{WaiteGTRI2025} to achieve tighter and less conservative safety guarantees. This work also differs from our approach, as it provides guarantees of reaching a particular state while in continuous space. %
While Fulton and Platzer \cite{Fulton_Platzer_2018} presented formal model-based safety proofs with reinforcement learning via runtime monitoring, this work doesn't explicitly address learned vision, neural perception, or the direct propagation of perception uncertainty into control safety. 
\cite{dean2020robust} utilizes a learned perception map to predict a linear function of the system's state to design a safe set and a robust controller for the closed-loop system. They demonstrate that, given suitable smoothness assumptions on both the perception map and the generative model, the parameters of the safe set can be learned through dense sampling of the state space.

\paragraph{\textbf{Safety Certificates.}}
Safety certificates, particularly Control Barrier Functions (CBFs), are widely used for run-time enforcement of safety~\cite{dean2021guaranteeing}. Differentiable and neural CBFs enable their integration into end-to-end pipelines to filter out unsafe actions and intervene to preserve safety, as in robotics~\cite{tong2023enforcing}. Works by Xiao et al.~\cite{Xiao2022WCAHR,Xiao2023TRMALR} evaluate the efficacy of CBFs in end-to-end vision-based autonomous driving and control. However, their formal guarantees primarily rely on the CBF layer itself and often do not explicitly model or formally bound vision-induced uncertainty. While these methods apply to autonomous driving with RGB camera images and may require explicit error bounds for estimators, they typically do not derive these bounds directly from learned vision or conformal prediction methods, unlike approaches such as Lindemann et al.~\cite{Lindemann2021RJTM}. More recently, Dawson et al.~\cite{dawson2022learning} proposed a learning-enabled hybrid controller with neural safety certificates, providing formal safety and liveness guarantees for systems using learned perception directly in observation space, thus not relying on explicit state estimation. An automated, formal, counterexample-based approach to synthesizing Barrier Certificates (BCs) as a neural network for the safety verification of continuous and hybrid dynamical models is introduced in~\cite{peruffo2020automated}.
\paragraph{\textbf{Analysis of Autonomous Systems with Learned Perception.}}
Formal analysis of autonomous systems that incorporate learned perception components is a growing area \cite{mitra2024formal}. Approaches such as DeepDECS~\cite{calinescu2023IMRPSV} integrate DNN verification and uncertainty quantification into formal (Markov model-based) controller synthesis to create formal, system-level safety-aware controllers. Similarly, probabilistic abstractions (via confusion matrices) and DNN-specific runtime guards for systems like airplane taxiing with vision DNNs have been explored \cite{PasareanuMGYICY23}. While Hsieh et al.~\cite{Hsieh2021LSJMM} approximated high-dimensional vision perception with low-dimensional surrogates for system-level verification, its safety is primarily verified empirically and through program analysis, lacking formal guarantees that connect actual perception errors directly to safety; it relies on empirical "precision" without formal coverage bounds. Another line of work involves scenario-based probabilistic compositional verification with symbolic abstractions, as seen in applications to airplanes and F1Tenth cars with DNN vision. However, these methods, like Watson et al.~\cite{Watson2025AGMP}, often focus on assumption generation for compositional verification rather than formally modeling vision uncertainty or providing finite-sample coverage guarantees or explicit set-based safety through conformal approaches.

\section{Conclusion}
\label{sec:conclusion}

In this paper, we study the problem of designing safe controllers (or policies) for discrete autonomous agents with imperfect perception. We present MDP formulations of such agents and propose a shield construction for such agents that restrict, at run-time, the actions available to the agent in each state as a function of state estimates. The construction combines the shield computed for a perfect-perception agent with conformalization of the perception component, with the aim of providing probabilistic, local safety guarantees with respect to finite horizons. Importantly, our construction does not require mathematical models of the complex process for generating observations from actual states and of the perception component. We also present a result proving \emph{global safety} of existing shield constructions for perfect-perception agents. We evaluate our constructions using the case study of an autonomous airplane taxiing system. Our experiments demonstrate that while the shield can improve safety, it comes at the cost of the system getting \emph{stuck}, i.e., reaching states where the shield is empty. In general, reaching stuck states is unavoidable under shielding, even with perfect perception. In future work, we plan to develop a proof of global safety for imperfect-perception agents using our conformal shield.

\bibliographystyle{splncs04}
\bibliography{bibfile}

\newpage
\appendix
\section{Proof of Theorem \ref{thm:global}}
\label{sec:proof_thm1}

\subsection{Definitions}
The following proof uses many common constructions from linear algebra to describe the progression of a controlled system. To describe these constructions we use the following notation. Let $(a_0,...,a_{n-1})$ denote a vector of length $n$ or equivalently, a matrix of dimension $n \times 1$. Let $\bracket{v,u}$ denote the inner product between vectors $v$ and $u$ of equal length. Let $A^\top$ denote the transpose of the matrix $A$. Adjoining two matrices $A$ and $B$, of dimension $N\times M$ and $M\times L$ respectively, represents their multiplication. Finally, we use $e_i$ to denote the $ith$ canonical basis vector, and $\mathbf{0}$ to denote the zero vector.

\subsection{Construction of a Faithful Model}

To reason about the probabilistic properties of the shielded control system we translate the previous formal constructions into a occupancy model of state transitions. This approach aligns with standard models of controlled Markovian systems (see, e.g., \cite{puterman1994mdp}, \cite{bertsekas2017dpoc}).

\begin{definition}
    Occupancy Vectors: For a control system $\Gamma_\mathrm{Perf}$, the state distribution at a given time is represented by an \textit{occupancy vector} $d \in \mathbb{R}^S$. Each entry $d_s$ gives the probability of being in state $s$ at that time-step. 
\end{definition}
Constructing an \textit{occupancy vector} relies on a fixed enumeration of the state space via a surjection $o: \mathbb{N} \to S$. Under this enumeration, each state $s \in S$ corresponds to a canonical basis vector $e_{o^{-1}(s)} \in \mathbb{R}^S$. Using this construction, the probability $d_s$ equates to the inner product $\bracket{e_{o^{-1}(s)},d}$.

Under perfect perception, the system dynamics are governed by the transition function $P: S \times A \times S \to [0,1]$, where $P(s,a,s')$ gives the probability of transitioning from state $s$ to $s'$ under action $a$. For a given policy $\pi: S \to A$, we define the corresponding \emph{transition matrix} $T_\pi \in \mathbb{R}^{S \times S}$ as:
\begin{align}
\label{eq:T-def}
    (T_\pi)_{i,j} = P(o(i), \pi(o(i)), o(j)).
\end{align}
This construction ensures that $T_\pi$ faithfully models the evolution of $\Gamma_\mathrm{Perf}$ using the policy $\pi$. Specifically, for any basis vector $e_i$, the probability of transitioning to $e_j$ in one time-step is:
\begin{align}
\label{eq:T-prop1}
    e_j^\top T_\pi e_i = P(o(i), \pi(o(i)), o(j)).
\end{align}

\begin{lemma}
\label{lem:perf:faithful_model}
The evolution of $\Gamma_\mathrm{Perf}$  can be described by a series of occupancy vectors and the transition matrix $T_\pi$. If $d$ represents the state distribution at time-step $i$, the distribution at the next time-step is:
\begin{align}
\label{eq:T-ev}
    d' = T_\pi d.
\end{align}
\end{lemma}
\begin{proof} Lemma \ref{lem:perf:faithful_model} follows from the construction of $T_\pi$. First, we may directly compute the probability of observing each state $s'$ at the next time-step, $d'_{s'}$, from the probability of being in each state $s$ at the current time, $d_{s}$. This formula is given by \eqref{eq:T_der:direct_prob}.

\begin{align}
\label{eq:T_der:direct_prob}
    d'_{s'} &= \sum_{s\in S} P(s,\pi(s),s') d_s 
\end{align}
From this basic statement of conditional probabilities, substituting eq. \eqref{eq:T-prop1} leads to eq. \eqref{eq:T_der:inter}.
\begin{align}
    d'_{s'} &= \sum_{s\in S} e_{o^{-1}(s')}^\top T_\pi e_{o^{-1}(s)} d_s \\
    \label{eq:T_der:inter}
    d'_{s'} &= e_{o^{-1}(s')}^\top T_\pi d
\end{align}
From \eqref{eq:T_der:inter} we may conclude \eqref{eq:T-ev} using the equality $d'_{s'}=\bracket{e_{o^{-1}(s')},d'}$. \hfill $\qed$
\end{proof}

\subsection{Modeling Constraints}

We will show how to model the satisfaction of the property "eventually unsafe" ($\psi=F^{\le n'} S_U$) in the occupancy model on the system $\Gamma_\mathrm{Perf}$ by modifying $\Gamma_\mathrm{Perf}$ and abstracting the modified system into the occupancy model. Using $S_U$, the set of unsafe states, define a system $\Gamma^\psi_\mathrm{Perf} = (S,A,\iota,P^\psi)$ by the transition function $P^\psi$, where $\delta$ is the Kronecker delta function.

\begin{align}
\label{eq:proof1_Ppsi_def}
    P^\psi(s,a,s') = \begin{cases}
        s \notin S_U & P(s,a,s')\\
        s \in S_U & \delta_{o^{-1}(s),o^{-1}(s')}
    \end{cases}
\end{align}

\begin{lemma}
\label{lem:perf:modified_gamma}
  The probability of a policy $\pi$ being unsafe according to the condition $\psi$ is preserved by the modification $\Gamma^\psi_\mathrm{Perf}$. Formally: \[
  \mathbb{P}_{\ppsysmdp,\pi}^{s}[\unsafe_{S_U,n'}(s)] = \mathbb{P}_{\Gamma_\mathrm{Perf}^\psi,\pi}^{s}[\unsafe_{S_U,n'}(s)] \] 
\end{lemma}
\begin{proof} A sequence of states $p=(s_0,...,s_{n'})$ which does not satisfy $\psi$ must not contain any elements of $S_U$. From \eqref{eq:proof1_Ppsi_def}, $P(s,a,s')$ and $P^\psi(s,a,s')$ must agree when $s,s'\notin S_U$.
Using these facts, and the definition of $\mathbb{P}_{\Gamma,\pi}$ from $\eqref{eq:path_prob}$, the probability measures $\mathbb{P}^{s_0}_{\Gamma_\mathrm{Perf},\pi}(p)$ and $\mathbb{P}^{s_0}_{\Gamma^\psi_\mathrm{Perf},\pi}(p)$ must  also agree. From this we can generalize $p$ to the set of all safe paths of length $n'$:

\[ \safe_{S_U,n'}(s) := \{(s_t)_{t=0}^{\infty} \in \paths_{\sysmdp,\kappa}(\hat{s})~|~\forall t \leq n'.~s_t \notin S_U\}\]
using \eqref{eq:path_set_prob} to guarantee $\mathbb{P}_{\Gamma_\mathrm{Perf},\pi}^{s_0}[\safe_{S_U,n'}]$ is equal to $\mathbb{P}_{\Gamma^\psi_\mathrm{Perf},\pi}^{s_0}[\safe_{S_U,n'}]$. As safety and unsafety are in direct correspondence,

\[\forall \Gamma.~ \mathbb{P}_{\Gamma,\pi}^{s_0}[\safe_{S_U,n'}(s)]=1-\mathbb{P}_{\Gamma,\pi}^{s_0}[\unsafe_{S_U,n'}(s)]\]
this implies Lemma \ref{lem:perf:modified_gamma}. \hfill $\qed$
\end{proof}

Lemma \ref{lem:perf:modified_gamma} shows how $\Gamma^\psi_\mathrm{Perf}$ is a \textit{$\psi$-preserving} transformation in the sense that $\Gamma^\psi_\mathrm{Perf}$ preserves the probability of satisfying $\psi$ on $\Gamma_\mathrm{Perf}$. Lemma \ref{lem:perf:modified_eventually} shows how the satisfaction of $\psi$ is more accessible in $\Gamma^\psi_\mathrm{Perf}$ than in $\Gamma_\mathrm{Perf}$, one only need consider the state occupancy at a given time-step rather than an entire path.

\begin{lemma}
\label{lem:perf:modified_eventually}
    The probability that a path $p \in \textbf{Paths}_{\Gamma^\psi_\mathrm{Perf},\pi}(s)$ satisfies $\psi$ is equal to the probability the final state of $p$ will be in $S_U$.
\end{lemma}
\begin{proof} The probability that a path $p \in \textbf{Paths}_{\Gamma^\psi_\mathrm{Perf},\pi}(s)$ satisfies $\psi$ is equal to the sum probability of paths which end in a state in $S_U$ plus the sum probability of paths which enter $S_U$ but end in $\lnot S_U$. Since the probability of leaving $S_U$ in $\Gamma^\psi_\mathrm{Perf}$ is zero, the sum probability of paths which enter $S_U$ but end in $\lnot S_U$ is zero. From this we may conclude Lemma \ref{lem:perf:modified_eventually}.\hfill $\qed$
\end{proof}

Using Lemmas \ref{lem:perf:modified_gamma} and \ref{lem:perf:modified_eventually} we may use the occupancy model of the system $\Gamma^\psi_\mathrm{Perf}$ and the probability that a path ends in $S_U$ to determine $\mathbb{P}_{\ppsysmdp,\pi}^{s}[\unsafe_{S_U,n}(s)]$. First, for a policy $\pi$, define a modified $T_\pi^\psi$ according to the transition function $P^\psi$ using eq. \eqref{eq:T-def}. Second, define a vector $C_{S_U}:\mathbb{R}^S$ \eqref{eq:csu_def} which, when applied through an inner product, $\bracket{C_{S_U},d}$, measures the probability $d\in \mathbb{R}^S$ is in an unsafe state \eqref{eq:proof1-Cprop}. 

\begin{align}
    \label{eq:csu_def}
    (C_X)_i &:= \begin{cases}
           o(i) \in X & 1 \\
           otherwise & 0\\
    \end{cases}~~~ (X \subseteq S) \\
    \bracket{C_X,d} &= \bracket{ \sum_{s\in X} e_{o^{-1}(s)},d} \nonumber \\
    &= \sum_{s\in X} \bracket{e_{o^{-1}(s)},d} \nonumber \\
    \label{eq:proof1-Cprop}
    &= \sum_{s\in X} d_s 
\end{align}
With the definitions of $C_{S_U}$ and $T^\psi_\pi$ we can express Lemma \ref{lem:perf:modified_eventually} more precisely \eqref{eq:lemma_mod_event}.

\begin{align}
\label{eq:lemma_mod_event}
    \mathbb{P}_{\Gamma^\psi_\mathrm{Perf},\pi}^{s}[\unsafe_{S_U,n'}(s)] = C_{S_U}^\top (T_\pi^\psi)^{n'} e_{o^{-1}(s)}
\end{align}

Lemma \ref{lem:perf:modified_gamma} allows us to drop $\psi$ from the left hand side of \eqref{eq:lemma_mod_event}, providing an occupancy interpretation of the measure $\mathbb{P}_{\ppsysmdp,\pi}^{s}[\unsafe_{S_U,n}(s)]$ through equality \eqref{eq:unsafe_redef}. To make later steps easier, we will instead use the equivalent expression \eqref{eq:basis_measure} to provide a lower bound on the probability of avoiding $S_U$.

\begin{align}
    \label{eq:unsafe_redef}
    \mathbb{P}_{\ppsysmdp,\pi}^{s}[\unsafe_{S_U,n'}(s)] &= C_{S_U}^\top (T^\psi_\pi)^{n'} e_{o^{-1}(s)} \\
    1-\mathbb{P}_{\ppsysmdp,\pi}^{s}[\safe_{S_U,n'}(s)] &= 1-C_{\lnot S_U}^\top (T^\psi_\pi)^{n'} e_{o^{-1}(s)} \nonumber \\
    \label{eq:basis_measure}
    \mathbb{P}_{\ppsysmdp,\pi}^{s}[\safe_{S_U,n'}(s)] &= C_{\lnot S_U}^\top (T^\psi_\pi)^{n'} e_{o^{-1}(s)}
\end{align}

Using eq. \eqref{eq:basis_measure} as a basis, we incrementally rebuild the previous constructions from Section \ref{sec:PP_MDP_Shielded} using this occupancy formulation. First, define a function $\xi_n$ \eqref{eq:xi_def} which represents the maximum probability of a path starting from $o(i)$ and avoiding $S_U$ within $n$ time-steps. Then, we may construct a vector $\Xi_n : \mathbb{R}^S$ \eqref{eq:Xi_def} from the optimums defined by $\xi_n$. In this way, the inner product $\bracket{\Xi_n,d}$ determines the probability of a system with occupancy vector $d$ avoiding $S_U$ in $n$ time-steps.

\begin{align}
    \label{eq:xi_def}
    \xi_n(e_i) &= \max_{\pi:S\rightarrow A} [C_{\lnot S_U} (T^\psi_{\pi})^n e_i] \\
    \label{eq:Xi_def}
    (\Xi_n)_i &= \xi_n(e_i)
\end{align}
From \eqref{eq:Xi_def} we may translate the definition of $\sigma_\psi(s,a)$ \eqref{eqn:sigma_def} into the occupancy model \eqref{eq:sigma_redef}. While the definition of $\Delta^n_\lambda$ remains the same, it is worth noticing that our perspective has changed to provide a lower bound as shown by eq. \eqref{eq:Delta_def}. In this and following equations $\lambda'=1-\lambda$.

\begin{align}
\label{eq:sigma_redef}
    1-\sigma^n_\psi(s,a) &= \Xi_{n}^\top P(s,a)\\
\label{eq:Delta_def}
    \Delta_\lambda^n(s) &= \{a\in A ~|~ 1-\sigma^n_\psi(s,a) \ge \lambda'\} 
\end{align}
To construct our proof it will be useful to define a restriction of the set $S$ to only those states which have safe actions according to $\Delta_\lambda^n$. We call this restriction $S_\Delta^{\lambda,n}$ and define it in eq. \eqref{eq:SDelta_def}.
\begin{align}
\label{eq:SDelta_def}
    S_\Delta^{\lambda,n} &= \{s \in S ~|~ \Delta_\lambda^n(s) \ne \emptyset \}
\end{align}

We would like to define a restriction of the type $S\rightarrow A$ which only includes control policies which obey shielding. One might do so by restricting the codomain to only the set $\Delta_\lambda^n$. However, this would prevent control policies from being total functions since $\forall s\notin S_\Delta^{\lambda,n}.~ \Delta_\lambda^n(s)=\emptyset$. A control policy $\pi$ must be a total function on $S$ for eq. \eqref{eq:T-def} to be well defined. To extend shielded controllers as total functions, we make no assumption about the action selected in states where the shield is empty.\footnote{Actions taken from unshielded states do not effect our analysis because of the \textit{no stuck} assumption.} This is demonstrated in \eqref{eq:bar_delta} and allows us to construct a dependent type $K_{\bar{\Delta}^n_\lambda}$ \eqref{eq:K_def} describing the behavior of shielded controllers in the loosest possible terms. 

\begin{align}
    \label{eq:bar_delta}
    \bar{\Delta}^n_\lambda &:= \begin{cases}
        s \in S_\Delta^{\lambda,n} & \Delta^n_\lambda \\
        otherwise & A
    \end{cases} \\
    \label{eq:K_def}
    K_{\bar{\Delta}^n_\lambda} &:= (s:S_\Delta^{\lambda,n}) \rightarrow \bar{\Delta}^n_\lambda(s)
\end{align}

Finally, we show how to construct a mask for a subset of states $X$. $M_X$ \eqref{eq:M_def} is a restriction of the identity matrix to only elements in $X$. As a function, $M_X$ passes through the probabilities of states within $X$ while filtering those outside of $X$. 
\begin{align}
\label{eq:M_def}
    (M_X)_{i,j} = \begin{cases}
        o(j) \in X & \delta_{i,j} \\
        otherwise & 0 
    \end{cases}
\end{align}

\subsection{Proof over Occupancy Construction}

\subsubsection{Overview}
\label{sec:proof_1_overview}

We would like to prove a greatest lower bound on the safety of a successive application of a series of shielded control policies. While the expression $(T^\psi_\pi)^n e_{o^{-1}(s)}$ represents the successive application of a single control policy $\pi$, we do not rely on the guarantee the same control policy is applied at each step. Rather, we merely assume the control policy applied at time-step $i$ is shielded: $\pi_i \in K_{\bar{\Delta}^n_\lambda}$. Using this weakened assumption we postulate the following where $f$ is some function of $\lambda$ and $n'$ satisfying \eqref{eq:occ_proof_postul}.

\begin{align}
\label{eq:occ_proof_postul}
    \forall i.~ \pi_i \in K_{\bar{\Delta}^n_\lambda} \rightarrow C_{\lnot S_U} \left( \Pi_{i=0}^{n'-1} T^\psi_{\pi_i} \right) e_{o^{-1}(s)} \ge f(\lambda,n')
\end{align}

The following proof finds $f^*$, the maximal $f$ which satisfies a modification of \eqref{eq:occ_proof_postul}. The modification is required, because, for $f^*(\lambda,n')$ to be greater than 0, it is necessary to make several additional assumptions. These are the \textit{no stuck} and \textit{initially safe} assumptions, which will be precisely defined below. To find $f^*$ we construct an inductive argument over $d_i$: the occupancy vector at time-step $i$. By proving that the \textit{no stuck} assumption implies an inductive hypothesis of the form $g(d_{i+1}) \ge \lambda' g(d_i)$  (where $g$ is a specific function of occupancy vectors that we define), and assuming an inductive base case, $g(d_0)=1$, we show that $g(d_{n'})=(\lambda')^{n'}$ for a path of length $n'$. Combining this with the property that $g(d)$ is a lower bound on the safety of $d$ allows us to conclude $f(\lambda,n') = (\lambda')^{n'}$. We show this lower bound is maximal $f^*(\lambda,n')=(\lambda')^{n'}$ in Section \ref{sec:proof_1_worst_case}.

\subsubsection{An Initial Inequality}

In the following steps we show how the use of a shielded control policy $\pi_i \in K_{\bar{\Delta}^n_\lambda}$ produces an inequality between the current occupancy state $d_i$ and the next occupancy state $d_{i+1}=T^\psi_{\pi_i} d_i$. 

\begin{prooftree}
    \AxiomC{$\forall i.~\pi_i : K_{\bar{\Delta}^n_\lambda}$  {\color{red}(p1.1)}}
    \UnaryInfC{$\forall s_\Delta \in S^{\lambda,n}_\Delta.~ \Xi_n^\top T^\psi_{\pi_i} e_{s_\Delta} \ge \lambda'$ {\color{red}(p1.2)}}
    \UnaryInfC{$\Xi_n^\top T_{\pi_i}^\psi M_{S^{\lambda,n}_\Delta} \ge \lambda' C_{S^{\lambda,n}_\Delta}^\top$ {\color{red}(p1.3)}}
    \UnaryInfC{$\Xi_n^\top T_{\pi_i}^\psi M_{S^{\lambda,n}_\Delta}d_i \ge \lambda' C_{S^{\lambda,n}_\Delta}^\top d_i$ {\color{red}(p1.4)}}
    \AxiomC{$\Xi_n^\top T_{\pi_i}^\psi d_i \ge \Xi_n^\top T_{\pi_i}^\psi M_{S^{\lambda,n}_\Delta}d_i$ {\color{red}(p1.5)}}
    \BinaryInfC{$\Xi_n^\top T_{\pi_i}^\psi d_i \ge \lambda' C_{S^{\lambda,n}_\Delta}^\top d_i$ {\color{red}(p1.6)}}
    \UnaryInfC{$\Xi_n^\top d_{i+1} \ge \lambda' C_{S^{\lambda,n}_\Delta}^\top d_i$ {\color{red}(p1.7)}}
\end{prooftree}
The first step, {\color{red}(p1.2)}, follows from {\color{red}(p1.1)} by applying the definitions of $K_{\bar{\Delta}^n_\lambda}$, $\bar{\Delta}^n_\lambda$, and $\sigma^n_\psi$. {\color{red}(p1.2)} describes a system of linear equations over all states in $S_\Delta^{\lambda,n}$. By extending this system with the inequality $\Xi_n^\top T_{\pi_i}^\psi  \mathbf{0} \ge \lambda' 0$ for all states not in $S_\Delta^{\lambda,n}$ we may construct the system of equations in  {\color{red}(p1.3)}. Right multiplication of {\color{red}(p1.3)} by $d_i$ produces {\color{red}(p1.4)}. Finally, {\color{red}(p1.5)} allows us to weaken the left hand side of {\color{red}(p1.4)} to find the definition of $d_{i+1}$ in {\color{red}(p1.6)}, which is replaced in {\color{red}(p1.7)}.

The inequality $\Xi_n^\top T_{\pi_i}^\psi d_i \ge \Xi_n^\top T_{\pi_i}^\psi M_{S^{\lambda,n}_\Delta}d_i$ from {\color{red}(p1.5)} is non obvious and deserves justification. In the following proof we can see how this inequality follows from a simple property of inner products interacting with masks {\color{red}(p2.1)}. The additional axiom in this proof $\bracket{\Xi_n^\top T^\psi_{\pi_i}, M_{\lnot S^{\lambda,n}_\Delta} d_i} \ge 0$ is an implication of the non-negative domain of probabilities.

\begin{prooftree}\smaller
    \AxiomC{$\bracket{y,z} = \bracket{y,M_X z} + \bracket{y,M_{\lnot X} z} $ {\color{red}(p2.1)}}
    \UnaryInfC{$\bracket{\Xi_n^\top T^\psi_{\pi_i}, d_i} = \bracket{\Xi_n^\top T^\psi_{\pi_i}, M_{S^{\lambda,n}_\Delta} d_i} + \bracket{\Xi_n^\top T^\psi_{\pi_i}, M_{\lnot S^{\lambda,n}_\Delta} d_i} $}
    \AxiomC{$\bracket{\Xi_n^\top T^\psi_{\pi_i}, M_{\lnot S^{\lambda,n}_\Delta} d_i} \ge 0$}
    \BinaryInfC{$\bracket{\Xi_n^\top T^\psi_{\pi_i}, d_i} \ge \bracket{\Xi_n^\top T^\psi_{\pi_i}, M_{S^{\lambda,n}_\Delta} d_i}$  }
    \UnaryInfC{$\Xi_n^\top T_{\pi_i}^\psi d_i \ge \Xi_n^\top T_{\pi_i}^\psi M_{S^{\lambda,n}_\Delta}d_i$ {\color{red} (p2.2)}}
\end{prooftree}
At the end of the {\color{red}(p1)}  we were left with the following inequality.

\begin{align}
\label{eq:proof1_init_ineq}
    \Xi_n^\top d_{i+1} \ge \lambda' C_{S^{\lambda,n}_\Delta}^\top d_i
\end{align}
To construct an inductive proof over $d_i$ we would like to match the left and right hand sides of \eqref{eq:proof1_init_ineq} to say $C_{S^{\lambda,n}_\Delta}^\top d_{i+1} \ge \lambda' C_{S^{\lambda,n}_\Delta}^\top d_i$. By assigning $g(d) := C_{S^{\lambda,n}_\Delta}^\top d$, this matching inequality will provide us with the inductive hypothesis of the form $g(d_i) \ge \lambda' g(d_{i+1})$ as promised in Section \ref{sec:proof_1_overview}.

\subsubsection{The Inductive Hypothesis}
\label{sec:proof_1_inductive_hyp}

Eq. \eqref{eq:proof1_init_ineq} leaves us pondering under what conditions $C_{S^{\lambda,n}_\Delta}^\top d \ge \Xi_n^\top d$. Comparing $C_{S^{\lambda,n}_\Delta}$ and $\Xi_n$, there are several cases where we can be assured that $(C_{S^{\lambda,n}_\Delta})_i \ge (\Xi_n)_i$. If $o(i)\in S^{\lambda,n}_\Delta$ then $(C_{S^{\lambda,n}_\Delta})_i=1$ which certainly must be greater than or equal to $(\Xi_n)_i$. Similarly, if $o(i) \in S_U$ then both $(C_{S^{\lambda,n}_\Delta})_i$ and $(\Xi_n)_i$ must be zero, again satisfying our inequality. The remaining cases of concern are therefore the remainder of $S$, which we will call $S^{\lambda,n}_\nabla$ \eqref{eq:S_nab_def}. This set consists of states which do not have safe actions according to $\Delta^n_\lambda$ but neither are they unsafe. Let us call these states \textit{stuck}. 
These states pose a major challenge to analysis because a shielded controller, as defined by eq. \eqref{eq:K_def}, may take any action from a state in $S_{\nabla}^{\lambda,n}$. This prevents the construction of any lower bound on the safety of actions taken from $S_{\nabla}^{\lambda,n}$. Furthermore, the probability of encountering stuck states cannot be bounded in the general case (Fig. \ref{fig:scarbros_demon}). Therefore, to construct a lower bound on the safety of the system as a whole, we make the \textit{no stuck} assumption: eq. \eqref{eq:no_stuck}. In words, this assumption requires that stuck states are not reachable from the initial state by a shielded controller.

\begin{align}
\label{eq:S_nab_def}
    S_{\nabla}^{\lambda,n} &= S \setminus (S_{\Delta}^{\lambda,n} \cup S_U) \\
\label{eq:no_stuck}
    \forall i.~ C_{S_\nabla^{\lambda,n}}^\top d_i &= 0
\end{align}
From the above discussion we may conclude $C_{S_\Delta^{\lambda,n}}^\top M_{\lnot S_\nabla^{\lambda,n}} \ge \Xi_n^\top M_{\lnot S_\nabla^{\lambda,n}}$. Combining this with the \textit{no stuck} assumption allows to conclude {\color{red}(p3.2)}.

\begin{prooftree}\smaller
    \AxiomC{$C_{S_\nabla^{\lambda,n}}^\top d_i = 0$ {\color{red}(p3.1)}}
    \UnaryInfC{$C_{S_\Delta^{\lambda,n}}^\top M_{\lnot S_\nabla^{\lambda,n}} d_i = C_{S_\Delta^{\lambda,n}}^\top  d_i$  $~~\Xi_n^\top M_{\lnot S_\nabla^{\lambda,n}} d_i = \Xi_n^\top  d_i$}
    \AxiomC{$C_{S_\Delta^{\lambda,n}}^\top M_{\lnot S_\nabla^{\lambda,n}} \ge \Xi_n^\top M_{\lnot S_\nabla^{\lambda,n}}$}
    \BinaryInfC{$C_{S_\Delta^{\lambda,n}}^\top d_i \ge \Xi_n^\top d_i$ {\color{red}(p3.2)}}
\end{prooftree}
The following summarizes the reasoning used above to convert the initial inequality {\color{red}(p1.7,p4.2)} into the inductive hypothesis {\color{red}(p4.2)} using the \textit{no stuck} assumption {\color{red}(p4.1)}.

\begin{prooftree}
    \AxiomC{$C_{S_\nabla^{\lambda,n}}^\top d_{i+1} = 0$ {\color{red}(p4.1)}}
    \UnaryInfC{$C_{S^{\lambda,n}_\Delta}^\top d_{i+1} \ge \Xi_n^\top d_{i+1}$}
    \AxiomC{$\Xi_n^\top d_{i+1} \ge \lambda' C_{S^{\lambda,n}_\Delta}^\top d_i$ {\color{red}(p4.2)}}
    \BinaryInfC{$C_{S^{\lambda,n}_\Delta}^\top d_{i+1} \ge \lambda' C_{S^{\lambda,n}_\Delta}^\top d_i$ {\color{red}(p4.2)}}
\end{prooftree}

\subsubsection{A Complete Proof}
Finally, an inductive proof of the safety of a perfect-perception model is given below. To begin the induction, we assume the initial state $\iota$ is in $S^{n,\lambda}_{\Delta}$ so that $C_{S^{n,\lambda}_{\Delta}}^\top d_0 = 1$, which we call the \textit{initially safe} assumption.

\begin{prooftree}
    \AxiomC{$\forall i.~ \pi_i : K_{\bar{\Delta}^n_\lambda}$ {\color{red}(p6.1)}}
    \UnaryInfC{$\Xi_n^\top d_{i+1} \ge \lambda C_{S^{\lambda,n}_\Delta}^\top d_i$ {\color{red}(p6.2)}}
    \AxiomC{$\forall i.~ C_{S_\nabla^{\lambda,n}}^\top d_i = 0$ {\color{red}(p6.3)}}
    \BinaryInfC{$\forall i.~ C_{S^{\lambda,n}_\Delta}^\top d_{i+1} \ge \lambda' C_{S^{\lambda,n}_\Delta}^\top d_i$ {\color{red}(p6.4)}}
    \AxiomC{$C_{S^{\lambda,n}_\Delta}^\top d_{0}=1$ {\color{red}(p6.5)}}
    \BinaryInfC{$C_{S^{\lambda,n}_\Delta}^\top d_{n'} \ge (\lambda')^{n'}$ {\color{red}(p6.6)}}
    \UnaryInfC{$C_{\lnot S_U}^\top d_{n'} \ge (\lambda')^{n'}$ {\color{red}(p6.7)}}
    \UnaryInfC{$C_{S_U}^\top d_{n'} \le 1-(1-\lambda)^{n'}${\color{red}(p6.8)}}
    \UnaryInfC{$C_{S_U}^\top \Pi_{i=0}^{n'} (T^\psi_{\pi_i}) d_0 \le 1-(1-\lambda)^{n'}${\color{red}(p6.9)}}
\end{prooftree}

Proof {\color{red}(p6)} summarizes the reasoning developed over the previous sections, combining each piece to give a full picture of the proof of Theorem \ref{thm:global}. To start, {\color{red}(p6.1)} implies {\color{red}(p6.2)} using {\color{red}(p1)}. Then, {\color{red}(p6.2)} and the no \textit{no stuck} assumption {\color{red}(p6.3)} are combined to produce {\color{red}(p6.4)} using {\color{red}(p4)}. From {\color{red}(p6.4)} we use the \textit{initially safe} assumption {\color{red}(p6.5)} and induction to conclude {\color{red}(p6.6)}. Then, {\color{red}(p6.6)} may be weakened using the property $S_{\Delta}^{\lambda,n} \subseteq \lnot S_U$ to conclude {\color{red}(p6.7)}. Finally, {\color{red}(p6.8)} follows from the law of total probabilities, and we replace $d_{n'}$ with its definition to conclude {\color{red}(p6.9)}.

{\color{red}(p6)} is a stronger statement than is required to prove Theorem \ref{thm:global}. Under the \textit{no stuck} and \textit{initially safe} conditions, the implication {\color{red}(p6.1)} $\rightarrow$ {\color{red}(p6.9)} may be weakened by applying the same policy $\pi \in K_{\Delta_\lambda^n}$ at each time-step \eqref{eq:thm1-eq}. 

\begin{align}
\label{eq:thm1-eq}
    \pi \in K_{\bar{\Delta}_\lambda^n} \rightarrow C_{S_U}^\top (T^\psi_{\pi})^n d_0 \le 1-(1-\lambda)^{n'}
\end{align}
Using \eqref{eq:unsafe_redef}, eq. \eqref{eq:thm1-eq} may be restated as \eqref{eq:thm1-eq2}.

\begin{align}
\label{eq:thm1-eq2}
    \pi \in K_{\bar{\Delta}_\lambda^n} \rightarrow \mathbb{P}_{\Gamma_{Perf,\pi}}[\unsafe_{S_U,n'}(s)] \le 1-(1-\lambda)^{n'}
\end{align}
Eq. \eqref{eq:thm1-eq2} is precisely the statement of Theorem \ref{thm:global} because the construction of $\Gamma_\mathrm{Perf}^\Delta$ restricts $\pi$ to the set of shielded policies $K_{\bar{\Delta}_\lambda^n}$.\hfill $\qed$

\subsection{Theoretical Case Studies}
\subsubsection{Worst Case Safety}
\label{sec:proof_1_worst_case}

It is relatively simple to demonstrate the lower bound on safety $f(\lambda,n') = (\lambda')^{n'}$ is maximal. Figure \ref{fig:worst_case_safety} provides an example of a system where repeated application of a shielded controller $\pi^* = \lambda. s_1 \rightarrow a_1$ achieves the lower bound on safety $(\lambda')^{n'}$. In this example $S_U=\{s_0\}$. Using the ordering $o=(s_1,s_0)$, we may compute $\Xi_n=(1,0)$ and $\Xi_n^\top T^\psi_{\pi^*}=(\lambda',0)$, where the choice of $n$ is inconsequential. This shows that $\pi^*\in K_{\bar{\Delta}_\lambda^n}$. Starting from $s_1$ and applying $\pi^*$ $n'$ times achieves the lower bound on safety \eqref{eq:worst_case}.
\begin{align}
\label{eq:worst_case}
    C_{\lnot S_U}^\top (T^\psi_{\pi^*})^n e_0 = (\lambda')^{n'}
\end{align}
\begin{figure}[t]
\centering
    \includegraphics[width=150pt]{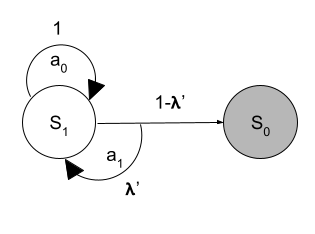}
    \caption{A simple example demonstrating worst case safety.}
    \label{fig:worst_case_safety}
\end{figure}

\subsubsection{Entrapment}
\label{sec:entrap}

\begin{figure}[t]
\centering
    \includegraphics[width=330pt]{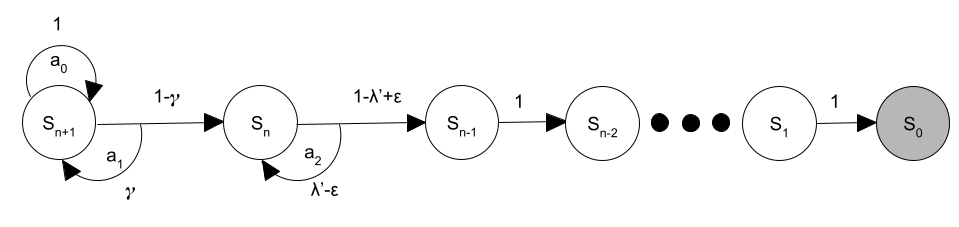}
    \caption{An example demonstrating the maximal stuck likelihood.}
    \label{fig:scarbros_demon}
    
\end{figure}

In the context of the relatively strong statement concerning shielded, perfect-perception models under the \textit{no stuck} assumption, one might wish to bound the probability of reaching a stuck state. Unfortunately, this cannot be done in the general case as demonstrated by the following case study. In Fig. \ref{fig:scarbros_demon} the probabilities of transitioning between states can be contrived so there exists a shielded controller which is certain to get stuck. For a shielded controller with shielding parameter $\lambda$ and look ahead $n$, we will show the worst case stuck probability by providing variable assignments which maximize $\bracket{C_{S^{\lambda,n}_\nabla},d_1}$ given $d_0=e_{o(s_{n+1})}$ and $S_U = \{s_0\}$. 

To start, for $\varepsilon > 0$, $S^{\lambda,n}_\Delta=\{s_{n+1}\}$. We would like to choose the minimal $\gamma$ so that $a_1\in \Delta_\lambda^n(s_{n+1})$ because minimizing $\gamma$ under this constraint will maximize the probability of reaching a stuck state. This is achieved by setting $\gamma=\varepsilon / (1-\lambda'-\varepsilon)$. If $a_1\in \Delta_\lambda^n(s_{n+1})$ then $\pi^* = \lambda.~ s_{n+1} \rightarrow a_1$ has type $K_{\bar{\Delta}_\lambda^n}$. These assignments imply the following.

\begin{align}
    \lim_{\varepsilon\rightarrow 0} C_{S^{\lambda,n}_\nabla}^\top T^\psi_{\pi^*} d_0 &= \lim _{\varepsilon\rightarrow 0} 1-\frac{\varepsilon}{1-\lambda'-\varepsilon} \nonumber \\
    &= 1
\end{align}

\section{Perfect-Perception MDP for Autonomous Taxiing}
\label{sec:prism_model}

\begin{lstlisting}[language={Prism}, numbers=left, rulesepcolor=\color{black}, rulecolor=\color{black}, breaklines=true, breakatwhitespace=true, firstnumber=1, firstline=1, 
caption={Perfect-Perception MDP},
label={mdp_prism5}]
mdp

const double dyn_suc=0.9;
const double dyn_fail=(1-dyn_suc)/2;

init  cte=0 & he=0 & cte_est=0 & he_est=0 & pc=1 & a=-1 & crash=0 endinit

module taxinet

	//  state
	cte : [-1..4];
	// cte=0: on centerline
	// cte=1: left of centerline
	// cte=2: right of centerline
	// cte=3: more left
 	// cte=4: more right
	// cte=-1: off runway ERROR STATE

	he: [-1..2];
	// he=0: no angle
	// he=1: negative
	// he=2: positive
	// he=-1: ERROR

	// estimate of the state from NN
	cte_est: [0..4];
	he_est: [0..2]; 
	

  	// program counter
	pc:[1..6]; 
	// 6 is "fail state"

	// actions issued by controller
	a:[-1..2];
	// a=-1: still to decide 
	// a=0: go straight
	// a=1: turn left
	// a=2: turn right

	// NN perfect perception 
	[] cte=0 & pc=1 -> 1.0: (cte_est'=0) & (pc'=2) + 
				 0.0: (cte_est'=1) & (pc'=2) + 
				 0.0: (cte_est'=2) & (pc'=2) + 
				 0.0: (cte_est'=3) & (pc'=2) + 
				 0.0: (cte_est'=4) & (pc'=2);
	[] cte=1 & pc=1 -> 0.0: (cte_est'=0) & (pc'=2) + 
				 1.0: (cte_est'=1) & (pc'=2) + 
				 0.0: (cte_est'=2) & (pc'=2) + 
				 0.0: (cte_est'=3) & (pc'=2) + 
				 0.0: (cte_est'=4) & (pc'=2);
	[] cte=2 & pc=1 -> 0.0: (cte_est'=0) & (pc'=2) + 
				 0.0: (cte_est'=1) & (pc'=2) + 
				 1.0: (cte_est'=2) & (pc'=2) + 
				 0.0: (cte_est'=3) & (pc'=2) + 
				 0.0: (cte_est'=4) & (pc'=2);
	[] cte=3 & pc=1 -> 0.0: (cte_est'=0) & (pc'=2) + 
				 0.0: (cte_est'=1) & (pc'=2) + 
				 0.0: (cte_est'=2) & (pc'=2) + 
				 1.0: (cte_est'=3) & (pc'=2) + 
				 0.0: (cte_est'=4) & (pc'=2);
	[] cte=4 & pc=1 -> 0.0: (cte_est'=0) & (pc'=2) + 
				 0.0: (cte_est'=1) & (pc'=2) + 
				 0.0: (cte_est'=2) & (pc'=2) + 
				 0.0: (cte_est'=3) & (pc'=2) + 
				 1.0: (cte_est'=4) & (pc'=2);


	[] he=0 & pc=2 -> 1.0: (he_est'=0) & (pc'=3) + 0.0: (he_est'=1) & (pc'=3) + 0.0: (he_est'=2) & (pc'=3); 	
	[] he=1 & pc=2 -> 0.0: (he_est'=0) & (pc'=3) + 1.0: (he_est'=1) & (pc'=3) + 0.0: (he_est'=2) & (pc'=3);
	[] he=2 & pc=2 -> 0.0: (he_est'=0) & (pc'=3) + 0.0: (he_est'=1) & (pc'=3) + 1.0: (he_est'=2) & (pc'=3);


	// Select actions
	/// CTE=0 ///
	[] cte_est=0 & he_est=0 & pc=3 -> (a'=0)&(pc'=4);	
	[] cte_est=0 & he_est=0 & pc=3 -> (a'=1)&(pc'=4);	
	[] cte_est=0 & he_est=0 & pc=3 -> (a'=2)&(pc'=4);	
	
	[] cte_est=0 & he_est=1 & pc=3 -> (a'=0)&(pc'=4);	
	[] cte_est=0 & he_est=1 & pc=3 -> (a'=1)&(pc'=4);	
	[] cte_est=0 & he_est=1 & pc=3 -> (a'=2)&(pc'=4);	

	[] cte_est=0 & he_est=2 & pc=3 -> (a'=0)&(pc'=4);	
	[] cte_est=0 & he_est=2 & pc=3 -> (a'=1)&(pc'=4);	
	[] cte_est=0 & he_est=2 & pc=3 -> (a'=2)&(pc'=4);	
	
	/// CTE=1 ///
	[] cte_est=1 & he_est=0 & pc=3 -> (a'=0)&(pc'=4);	
	[] cte_est=1 & he_est=0 & pc=3 -> (a'=1)&(pc'=4);	
	[] cte_est=1 & he_est=0 & pc=3 -> (a'=2)&(pc'=4);	
	
	[] cte_est=1 & he_est=1 & pc=3 -> (a'=0)&(pc'=4);	
	[] cte_est=1 & he_est=1 & pc=3 -> (a'=1)&(pc'=4);	
	[] cte_est=1 & he_est=1 & pc=3 -> (a'=2)&(pc'=4);	

	[] cte_est=1 & he_est=2 & pc=3 -> (a'=0)&(pc'=4);	
	[] cte_est=1 & he_est=2 & pc=3 -> (a'=1)&(pc'=4);	
	[] cte_est=1 & he_est=2 & pc=3 -> (a'=2)&(pc'=4);	
	
	/// CTE=2 ///
	[] cte_est=2 & he_est=0 & pc=3 -> (a'=0)&(pc'=4);	
	[] cte_est=2 & he_est=0 & pc=3 -> (a'=1)&(pc'=4);	
	[] cte_est=2 & he_est=0 & pc=3 -> (a'=2)&(pc'=4);	
	
	[] cte_est=2 & he_est=1 & pc=3 -> (a'=0)&(pc'=4);	
	[] cte_est=2 & he_est=1 & pc=3 -> (a'=1)&(pc'=4);	
	[] cte_est=2 & he_est=1 & pc=3 -> (a'=2)&(pc'=4);	

	[] cte_est=2 & he_est=2 & pc=3 -> (a'=0)&(pc'=4);	
	[] cte_est=2 & he_est=2 & pc=3 -> (a'=1)&(pc'=4);	
	[] cte_est=2 & he_est=2 & pc=3 -> (a'=2)&(pc'=4);	
	
	/// CTE=3 ///
	[] cte_est=3 & he_est=0 & pc=3 -> (a'=0)&(pc'=4);	
	[] cte_est=3 & he_est=0 & pc=3 -> (a'=1)&(pc'=4);	
	[] cte_est=3 & he_est=0 & pc=3 -> (a'=2)&(pc'=4);	
	
	[] cte_est=3 & he_est=1 & pc=3 -> (a'=0)&(pc'=4);	
	[] cte_est=3 & he_est=1 & pc=3 -> (a'=1)&(pc'=4);	
	[] cte_est=3 & he_est=1 & pc=3 -> (a'=2)&(pc'=4);	

	[] cte_est=3 & he_est=2 & pc=3 -> (a'=0)&(pc'=4);	
	[] cte_est=3 & he_est=2 & pc=3 -> (a'=1)&(pc'=4);	
	[] cte_est=3 & he_est=2 & pc=3 -> (a'=2)&(pc'=4);	
	
	/// CTE=4 ///
	[] cte_est=4 & he_est=0 & pc=3 -> (a'=0)&(pc'=4);	
	[] cte_est=4 & he_est=0 & pc=3 -> (a'=1)&(pc'=4);	
	[] cte_est=4 & he_est=0 & pc=3 -> (a'=2)&(pc'=4);	
	
	[] cte_est=4 & he_est=1 & pc=3 -> (a'=0)&(pc'=4);	
	[] cte_est=4 & he_est=1 & pc=3 -> (a'=1)&(pc'=4);	
	[] cte_est=4 & he_est=1 & pc=3 -> (a'=2)&(pc'=4);	

	[] cte_est=4 & he_est=2 & pc=3 -> (a'=0)&(pc'=4);	
	[] cte_est=4 & he_est=2 & pc=3 -> (a'=1)&(pc'=4);	
	[] cte_est=4 & he_est=2 & pc=3 -> (a'=2)&(pc'=4);

	// action selection stochasticity
	[] a=0 & pc=4 -> dyn_suc: (a'=0) & (pc'=5) + dyn_fail: (a'=1) & (pc'=5) + dyn_fail: (a'=2) & (pc'=5);
	[] a=1 & pc=4 -> dyn_suc: (a'=1) & (pc'=5) + dyn_fail: (a'=0) & (pc'=5) + dyn_fail: (a'=2) & (pc'=5);
	[] a=2 & pc=4 -> dyn_suc: (a'=2) & (pc'=5) + dyn_fail: (a'=0) & (pc'=5) + dyn_fail: (a'=1) & (pc'=5);
	
	// airplane and environment: 
	[] he=1 & a=1 & pc=5 -> 1 : (he'=-1) & (pc'=6); // error
	[] he=2 & a=2 & pc=5 -> 1 : (he'=-1) & (pc'=6); // error

	[] cte=0 & he=0 & a=0 & pc=5 -> 1 : (pc'=1) & (a'=-1);
	[] cte=0 & he=0 & a=1 & pc=5 -> 1 : (cte'=1) & (he'=1) & (pc'=1) & (a'=-1);
	[] cte=0 & he=0 & a=2 & pc=5 -> 1 : (cte'=2) & (he'=2) & (pc'=1) & (a'=-1);

	[] cte=0 & he=1 & a=0 & pc=5 -> 1 : (cte'=1) & (pc'=1) & (a'=-1);
	[] cte=0 & he=1 & a=2 & pc=5 -> 1 : (he'=0) & (pc'=1) & (a'=-1);

	[] cte=0 & he=2 & a=0 & pc=5 -> 1 : (cte'=2) & (pc'=1) & (a'=-1);
	[] cte=0 & he=2 & a=1 & pc=5 -> 1 : (he'=0) & (pc'=1) & (a'=-1);

	// left-side dynamics:
	[] cte=1 & he=0 & a=0 & pc=5 -> 1 : (pc'=1) & (a'=-1);
	[] cte=1 & he=0 & a=1 & pc=5 -> 1 : (cte'=3) & (he'=1) & (pc'=1) & (a'=-1);
	[] cte=1 & he=0 & a=2 & pc=5 -> 1 : (cte'=0) & (he'=2) & (pc'=1) & (a'=-1);

	[] cte=1 & he=1 & a=0 & pc=5 -> 1 : (cte'=3) & (pc'=1) & (a'=-1);
	[] cte=1 & he=1 & a=2 & pc=5 -> 1 : (he'=0) & (pc'=1) & (a'=-1);

	[] cte=1 & he=2 & a=0 & pc=5 -> 1 : (cte'=0) & (pc'=1) & (a'=-1);
	[] cte=1 & he=2 & a=1 & pc=5 -> 1 : (he'=0) & (pc'=1) & (a'=-1);

	[] cte=3 & he=0 & a=0 & pc=5 -> 1 : (pc'=1) & (a'=-1);
	[] cte=3 & he=0 & a=1 & pc=5 -> 1 : (cte'=-1) & (pc'=6); //error
	[] cte=3 & he=0 & a=2 & pc=5 -> 1 : (cte'=1) & (he'=2) & (pc'=1) & (a'=-1);

	[] cte=3 & he=1 & a=0 & pc=5 -> 1 : (cte'=-1) & (pc'=6); //error
	[] cte=3 & he=1 & a=2 & pc=5 -> 1 : (he'=0) & (pc'=1) & (a'=-1);

	[] cte=3 & he=2 & a=0 & pc=5 -> 1 : (cte'=1) & (pc'=1) & (a'=-1);
	[] cte=3 & he=2 & a=1 & pc=5 -> 1 : (he'=0) & (pc'=1) & (a'=-1);

	// right-side dynamics:
	[] cte=2 & he=0 & a=0 & pc=5 -> 1 : (pc'=1) & (a'=-1);
	[] cte=2 & he=0 & a=1 & pc=5 -> 1 : (cte'=0) & (he'=1) & (pc'=1) & (a'=-1);
	[] cte=2 & he=0 & a=2 & pc=5 -> 1 : (cte'=4) & (he'=2) & (pc'=1) & (a'=-1);

	[] cte=2 & he=1 & a=0 & pc=5 -> 1 : (cte'=0) & (pc'=1) & (a'=-1);
	[] cte=2 & he=1 & a=2 & pc=5 -> 1 : (he'=0) & (pc'=1) & (a'=-1);

	[] cte=2 & he=2 & a=0 & pc=5 -> 1 : (cte'=4) & (pc'=1) & (a'=-1);
	[] cte=2 & he=2 & a=1 & pc=5 -> 1 : (he'=0) & (pc'=1) & (a'=-1);

	[] cte=4 & he=0 & a=0 & pc=5 -> 1 : (pc'=1) & (a'=-1);
	[] cte=4 & he=0 & a=1 & pc=5 -> 1 : (cte'=2) & (he'=1) & (pc'=1) & (a'=-1);
	[] cte=4 & he=0 & a=2 & pc=5 -> 1 : (cte'=-1) & (pc'=6); //error

	[] cte=4 & he=1 & a=0 & pc=5 -> 1 : (cte'=2) & (pc'=1) & (a'=-1);
	[] cte=4 & he=1 & a=2 & pc=5 -> 1 : (he'=0) & (pc'=1) & (a'=-1);

	[] cte=4 & he=2 & a=0 & pc=5 -> 1 : (cte'=-1) & (pc'=6); //error
	[] cte=4 & he=2 & a=1 & pc=5 -> 1 : (he'=0) & (pc'=1) & (a'=-1);

endmodule
\end{lstlisting}

\end{document}